\documentclass[11pt]{article}

\usepackage{fullpage}

\newcommand{\cA}{\mathcal A}
\newcommand{\cAloc}{\mathcal A^\mathrm{loc}}
\newcommand{\cAal}{\mathcal A^\mathrm{al}}
\newcommand{\cF}{\mathcal F}
\newcommand{\cS}{\mathcal S}
\newcommand{\cSs}{\mathcal S_\mathrm{sym}}
\newcommand{\bbZZ}{\mathbb{Z}}
\newcommand{\calL}{\mathcal{L}}
\usepackage{amsmath,amssymb,amsthm,dsfont,mathrsfs}

\newcommand{\up}{\uparrow}

\newcommand{\caA}{{\mathcal A}}

\newcommand{\caD}{{\mathcal D}}

\newcommand{\caF}{{\mathcal F}}

\newcommand{\caH}{{\mathcal H}}
\newcommand{\caI}{{\mathcal I}}
\newcommand{\caJ}{{\mathcal J}}

\newcommand{\caL}{{\mathcal L}}

\newcommand{\caN}{{\mathcal N}}

\newcommand{\caP}{{\mathcal P}}

\newcommand{\caU}{{\mathcal U}}

\newcommand{\bbC}{{\mathbb C}}

\newcommand{\bbN}{{\mathbb N}}

\newcommand{\bbR}{{\mathbb R}}
\newcommand{\bbS}{{\mathbb S}}

\newcommand{\bbZ}{{\mathbb Z}}

\newcommand{\iu}{\mathrm{i}}

\newcommand{\str}{^{*}}
\newcommand{\ep}[1]{\mathrm{e}^{#1}}

\newcommand{\dd}{\,\mathrm{d}}
\newcommand{\tr}{\mathrm{tr}}
\newcommand{\Tr}{\mathrm{Tr}}

\newcommand{\norm}[1]{\left\Vert #1 \right\Vert}

\newcommand{\Ad}[1]{\operatorname{Ad}_{#1}}

\DeclareMathOperator*{\slim}{s-lim}

\newcommand{\ee}{{\mathrm e}}
\newcommand{\Idi}{{\mathrm i}}
\usepackage[dvipsnames]{xcolor}

\usepackage{tikz}
\usetikzlibrary{positioning, arrows.meta}

\newcommand{\qa}{{\rm qa}}
\newcommand\myshade{85}
\colorlet{mylinkcolor}{violet}
\colorlet{mycitecolor}{YellowOrange}
\colorlet{myurlcolor}{Aquamarine}

\usepackage[pdftex,%
colorlinks=true,linkcolor=mylinkcolor!\myshade!black,citecolor=mycitecolor!\myshade!black,urlcolor=myurlcolor!\myshade!black
]
{hyperref}

\usepackage{braket}
\usepackage{environ}
\usepackage[nameinlink]{cleveref}
\NewEnviron{malign}{%
	\begin{align}\begin{split}
			\BODY
	\end{split}\end{align}
}

\newcommand{\ip}[2]{\langle #1, #2 \rangle}

\newcommand{\eq}[1]{\begin{align*}#1\end{align*}}
\newcommand{\eql}[1]{\begin{align}#1\end{align}}

\newcommand{\Id}{\mathds{1}}

\newcommand{\szpan}{\operatorname{span}}
\newcommand{\NN}{\mathbb{N}}
\newcommand{\RR}{\mathbb{R}}

\newcommand{\CC}{\mathbb{C}}

\newcommand{\ZZ}{\mathbb{Z}}

\newcommand{\ii}{\operatorname{i}}
\newcommand{\calA}{\mathcal{A}}
\newcommand{\calB}{\mathcal{B}}

\newcommand{\calN}{\mathcal{N}}

\newcommand{\calD}{\mathcal{D}}

\newcommand{\calU}{\mathcal{U}}
\newcommand{\calH}{\mathcal{H}}

\newcommand{\calK}{\mathcal{K}}

\newcommand{\calJ}{\mathcal{J}}
\newcommand{\calP}{\mathcal{P}}

\newcommand{\br}[1]{\left(#1\right)}

\definecolor{ct_black}{HTML}{000000}
\definecolor{ct_orange}{HTML}{ED872D}
\definecolor{ct_purple}{HTML}{7A68A6}
\definecolor{ct_blue}{HTML}{348ABD}
\definecolor{ct_turquoise}{HTML}{188487}
\definecolor{ct_red}{HTML}{E32636}
\definecolor{ct_pink}{HTML}{CF4457}
\definecolor{ct_green}{HTML}{467821}

\definecolor{ct2_green}{HTML}{9FF781}
\definecolor{ct2_green_dark}{HTML}{088A08}

\theoremstyle{plain}
\newtheorem{thm}{\protect\theoremname}[section]
\theoremstyle{plain}
\newtheorem{lem}[thm]{\protect\lemmaname}
\theoremstyle{plain}

\theoremstyle{plain}
\newtheorem{prop}[thm]{\protect\propositionname}
\theoremstyle{plain}

\newtheorem{conj}[thm]{\protect\conjname}

\usepackage{mathtools}
\DeclarePairedDelimiter\abs{\lvert}{\rvert}

\theoremstyle{remark}
\newtheorem{rem}[thm]{\protect\remarkname}
\theoremstyle{definition}
\newtheorem{defn}[thm]{\protect\definitionname}%[section]
\theoremstyle{plain}
\newtheorem{example}[thm]{\protect\examplename}%[section]

\providecommand{\assumptionname}{Assumption}
\providecommand{\claimname}{Claim}
\providecommand{\corollaryname}{Corollary}
\providecommand{\definitionname}{Definition}
\providecommand{\lemmaname}{Lemma}
\providecommand{\propositionname}{Proposition}
\providecommand{\remarkname}{Remark}
\providecommand{\theoremname}{Theorem}
\providecommand{\examplename}{Example}
\providecommand{\conjname}{Conjecture}

\crefname{section}{Section}{Sections}
\crefname{example}{Example}{Examples}
\crefname{appendix}{Appendix}{Appendices}
\crefname{figure}{Figure}{Figures}
\crefname{assumption}{Assumption}{Assumptions}
\crefname{thm}{Theorem}{Theorems}
\crefname{lem}{Lemma}{Lemmas}
\crefformat{equation}{(#2#1#3)}
\crefname{table}{Table}{Tables}

\crefrangelabelformat{equation}{(#3#1#4--#5#2#6)}

\crefmultiformat{equation}{(#2#1#3}{, #2#1#3)}{#2#1#3}{#2#1#3}
\Crefmultiformat{equation}{(#2#1#3}{, #2#1#3)}{#2#1#3}{#2#1#3}

\newcommand{\plane}{\mathbb{P}}
\newcommand{\D}{\gamma}

\usepackage{authblk}

\title{The index of a pair of pure states and \\ the interacting integer quantum Hall effect}

\author[1]{Sven Bachmann}
\affil[1]{Department of Mathematics, The University of British Columbia, 1984 Mathematics Road, Vancouver BC, Canada V6T 1Z2}
\author[2]{Jacob Shapiro}
\affil[2]{Department of Mathematics, Princeton University, Fine Hall, Princeton, NJ 08540, USA}
\author[3]{Cl\'ement Tauber}
\affil[3]{CEREMADE, CNRS, Universit\'e Paris-Dauphine, Universit\'e PSL, 75016 Paris, France}

\begin{document}

\maketitle

\abstract{ We introduce the index $\caN(\omega_1,\omega_2)$ of a pair of pure states on a unital C*-algebra, which is a generalization of the notion of the index of a pair of projections on a Hilbert space.
We then show that the Hall conductance associated with an invertible state $\omega$ of a two-dimensional interacting electronic system which is symmetric under $U(1)$ charge transformation may be written as the index $\calN(\omega,\omega^D)$, where $\omega^D$ is obtained from $\omega$ by inserting a unit of magnetic flux. This  exhibits the integrality and continuity properties of the Hall conductance in the context of general topological features of $\calN$.}

%%%%%%%%%%%%%%%%%%%%%%%%%%%%%%%%%%%%%%

\section{Introduction}

The index of a pair of projections introduced in~\cite{Kato1966,BDF1973,AvronSeilerSimon_Charge} and explored further in~\cite{AvronSeilerSimon} is an integer-valued index that can be associated with two projections $P,Q$ on a Hilbert space whenever their difference is compact: \eq{
\operatorname{index}(P,Q) = \dim \br{ \operatorname{im}(P)\cap\ker(Q) } - \dim \br{ \operatorname{im}(Q)\cap\ker(P) } \in \ZZ\,.
} It is `topological' in the sense that it is invariant under compact and norm-continuous deformations. In the case where one of the projections is a unitary conjugate of the other, $Q = U^\ast P U$, the index reduces to a Fredholm index of $PUP+P^\perp$.

Since its introduction, this notion has been intimately connected to the quantum Hall effect. It is indeed one of the possible expressions for the Hall conductance, when the currents are driven by the adiabatic increase of a magnetic flux through the two-dimensional electron gas. The physical meaning of $P$ is that of a Fermi projection whenever the Fermi energy lies in a spectral gap or a mobility gap~\cite{AizenmanGraf}. In particular, this picture is valid in a non-interacting setting where the many-body ground state reduces to a one-body projection, see~\cite{graf2007aspects} for an overview and further references. This concept was eventually generalized to yield the index of all entries of the Kitaev periodic table of topological insulators \cite{KatsuraKoma2016_Z2Index,KatsuraKoma2018_NCIndexThm}.

There have been various attempts at generalizing the index of a pair of projections to an interacting setting where the state is not simply given by a projection in Hilbert space, most notably~\cite{bachmann2020many} in an arbitrarily large finite volume followed by~\cite{kapustin2020hall} in the infinite volume setting, where quantization can be proved under the assumption of invertibility of the initial state. In both cases, the attention is on the Hall effect.

In this work, we introduce a generalization of the above indices which is defined in a completely abstract setting. It is associated with two pure states $(\omega_1,\omega_2)$ of a C*-algebra $\caA$ that are related by an inner automorphism, in the presence of a $U(1)$ symmetry. Specifically, when $\omega_2 = \omega_1\circ\mathrm{Ad}_u$, where $u\in\caA$ is unitary, the index is given by
\begin{equation}
    \calN_\rho(\omega_1,\omega_2) := \ii\omega_1\br{u^\ast \delta^\rho (u)},
\end{equation}
where $\delta^\rho$ is the generator of the $U(1)$-symmetry. We prove, among other things, that the index is integer valued, operator-norm-continuous, and invariant under deformations by symmetric automorphisms. We also show that if the algebra is chosen to be the CAR algebra describing fermions and if the two states are quasi-free, then our index reduces to the index of a pair of projections. Therefore, this new index of a pair of pure states generalizes~\cite{AvronSeilerSimon}, as well as the noncommutative geometric approach of~\cite{BellissardNCG}.

We then turn to the quantum Hall effect. Echoing~\cite{kapustin2020hall}, we prove that the piercing of a magnetic flux from $0$ to $2\pi$ starting from an \emph{invertible} state corresponds to a situation where the many-body index is well-defined. The technical part here uses adiabatic flux insertion \`a la Laughlin~\cite{Laughlin} extended to the many-body setting in~\cite{BachmannBolsRahnama24}. Unlike there, we focus on charge conservation rather than time-reversal symmetry. In this context, the many-body index equals the charge deficiency of the final state with respect to the initial state, and so to the Hall conductance by the Laughlin argument. This places the Hall index into a very general C*-algebraic framework which is valid for both interacting and non-interacting fermionic systems, generalizing previous expressions~\cite{bachmann2020many, kapustin2020hall} and placing them in a functional analytic framework which complements other approaches using algebraic topology~\cite{KSNoether,KapustinArtymowicz}. 

We point out immediately that the invertibility assumption appears in two roles that are very distinct from each other. It first plays a role as a tool that allows us to ensure that the state is the gapped ground state of a so-called parent Hamiltonian. Secondly, it is crucial in concluding that the defect state obtained after flux insertion is locally comparable to the initial state, namely that the two differ only in the vicinity of the puncture and not at infinity. While the first role could be bypassed in a physical setting where the state is given as a gapped ground state, the second one appears fundamental. Indeed, it is that very assumption that ensures that the state has integer quantum Hall conductance as opposed to fractional conductance. In the latter case, the initial state is expected to have non-trivial superselection sectors corresponding to anyonic excitations (see~\cite{ogataCategory} for one possible meaning of these terms and~\cite{Froehlich} for an overview of anyons in the fractional quantum Hall effect) and these cannot be realized upon an invertible state~\cite{KitaevInvertible,Triviality}.

Let us briefly discuss other approaches to (integer) quantization in the interacting quantum Hall effect. In the single particle picture, charge deficiency, charge transport and linear response coincide \cite{Laughlin,AvronSeilerSimon_Charge,graf2007aspects}. This equivalence continues to hold in the interacting picture~\cite{bachmann2020many,kapustin2020hall}, where quantization of the Hall conductance was proved in~\cite{hastings2015quantization, bachmann2018quantization}, see also~\cite{bachmann18,monaco19} for the validity of linear response. While the above assume the presence of a gap, this assumption is proved to hold in a perturbative setting in~\cite{giuliani17}, while it is not needed if one considers non-equilibrium almost steady states~\cite{wesle25,teufel25}.

The paper is organized as follows. In \cref{sec:abstract_index}, we describe the general algebraic setting, introduce the notion of two pure states being $\rho$-locally-comparable (see \cref{def:alpha comp}) and define the index $\caN_\rho(\omega_1,\omega_2)$. If $\rho_t$ is a $U(1)$-symmetry, then
$$\caN_\rho(\omega_1,\omega_2)\in \ZZ.$$
In \cref{subsec:IPP}, we consider the fermionic algebra $\cA=\mathrm{CAR}(\calH)$ on a one-particle Hilbert space $\calH$. For $P,Q$ two orthogonal projections on $\mathcal H$ and $\omega_P, \omega_Q$ their corresponding quasi-free states on $\cA$, we show that if $P-Q$ is trace class then
$$
\caN_\rho(\omega_P,\omega_Q) = \mathrm{index}(P,Q),
$$
the index of a pair of projections. \cref{sec:IQHE} focusses on the quantum Hall effect in an interacting setting where the two states are an initial invertible state and a defect state obtained from the initial one by inserting a unit of flux at the origin. Then the charge deficiency, and therefore the Hall conductance, is given by
$$
\calN_{\rho\otimes\mathrm{id}} (\hat \omega, \hat\omega^D) \in \ZZ
$$
where $\hat\cdot$ corresponds to a stacking operation associated to invertible states. The fact that we pick $\rho\otimes\mathrm{id}$ rather than the stacked $\hat{\rho}$ means that we measure only the charge transported in the original system, rather than in the full, stacked system. We then show in \cref{sec:free} that this index covers gapped ground states of interacting Hamiltonians as well, and that in the quasi-free case it coincides with a single-particle spectral flow.

%%%%%%%%%%%%%%%%%%%%%%%%%%%%%%%%%%%%%%%%%%%%%%%%%%%%%%%%%%%%%%%
\section{Abstract index theory \label{sec:abstract_index}}
\subsection{The index of a pair of pure states}
    
Let $\calA$ be a \emph{unital} C*-algebra. In later sections we shall apply the theory to $\calA=\operatorname{CAR}(\calH)$ for some separable Hilbert space $\calH$ but in fact throughout this section no assumptions on $\calA$ will be made. We denote by $\calU(\calA)$ the set of unitary elements of $\calA$. A state $\omega : \cA \to \mathbb C$ is a positive (i.e., $\omega(a\str a)\geq0$ for all $a\in\calA$) linear functional on $\cA$ with $\omega(\Id)=1$; the space of all states shall be denoted $\cS(\cA)$. 
The condition $\omega(\Id)=1$ is equivalent to the normalization $\|\omega \| =1$, where the norm is given by
\eq{
        \norm{\omega} := \sup\br{\Set{\abs{\omega(a)} | a\in\calA : \norm{a}\leq 1}}\,.
    }
In other words, all states are on the unit sphere of $\caA\str$ and so $\cS(\cA)$ is a weakly*-compact and convex subset of $\cA\str$. The extreme points of $\cS(\cA)$ are called pure states and their collection is denoted by $\caP(\caA)$.

    \begin{defn}[locally-comparable pair of pure states]\label{def:comp}
        Let $\omega_1,\omega_2\in\calP(\calA)$ be a pair of pure states. We say that the pair $(\omega_1,\omega_2)$ is \emph{locally-comparable} iff there exists some $u\in\calU(\calA)$ such that \eql{\label{eq:locally comparable pure states}
            \omega_2 = \omega_1 \circ \Ad{u}\,.
        } 
    \end{defn}
Here and in the sequel, $\Ad{u}$ denotes the automorphism $\Ad{u}(a) = u\str a u$. In other words, $\omega_1$ and $\omega_2$ are locally comparable iff they are inner-automorphism equivalent.

    \begin{rem}
        The intuition we have in mind for two pure states to be locally comparable is that their difference is ``compact'' in a vague sense. Since \eq{
         \omega_2 -\omega_1 = \omega_1\circ \br{\Ad{u}-\mathrm{id}}\,,
    } if $\calA=\operatorname{CAR}(\calH)$ for some Hilbert space $\calH$, then $\calA$ may be considered as the operator norm limit of finite-rank (but not necessarily quadratic) observables, and as such, since $u\in\calU(\calA)$, we should consider $u-\Id$ and hence also $\Ad{u}-\mathrm{id}$ to be (the many-body analog of) compact. Then by the ideal property, the whole expression $\omega_1\circ\br{\Ad{u}-\mathrm{id}}$ is.
    \end{rem}

    Back to the case of a general C*-algebra $\calA$, when two pure states $\omega_1,\omega_2$ are locally-comparable, we want to measure how different they are, with the expectation that if they are path-connected (in $\calP(\calA)$) they should not be different at all. This is the case when $u\in\calU_0(\calA)$, i.e., when $u$ may be continuously deformed to $\Id$, whence we get a continuous deformation of $\omega_2$ to $\omega_1$. However, there are C*-algebras where $\pi_0(\calU(\calA))\neq\Set{0}$ (for example $C(\bbS^1)$). It will turn out that we are not quite interested in such obstructions. Rather, in many interesting cases (such as UHF algebras), $\pi_0(\calU(\calA))=\Set{0}$. Then, to gauge a topological obstruction we impose a further symmetry constraint.

    Let $\Set{\rho_t}_{t\in\RR}\subseteq\operatorname{Aut}(\calA)$ be a strongly-continuous one-parameter group of  *-automorphisms. Its generator $\delta^\rho$ given by \eq{
        \delta^\rho := \slim_{t\to0^+}\frac{1}{t}\br{\rho_t-\mathrm{id}}
    } is a *-derivation. In general $\rho_t$ is not expected to be inner and $\delta^\rho$ is not expected to be bounded. As such, we must consider $\delta^\rho$ together with its domain \eq{
    \calD(\delta^\rho) \equiv \Set{a\in\calA | \lim_{t\to0^+}\frac{1}{t}\br{\rho_t(a)-a}\text{ exists in }\calA}\,.
    }

     With this, we refine the notion of locally-comparable pure states as follows.

    \begin{defn}[$\rho$-locally-comparable pair of pure states]\label{def:alpha comp}
        Let $\rho_t$ be as above and $(\omega_1,\omega_2)$ be a pair of pure states. We say that this pair is $\rho$-locally-comparable iff there exists some $u\in \calU(\calA)\cap\calD(\delta^\rho)$ such that \eq{
        \omega_2 = \omega_1 \circ \Ad{u}\,.
        }
    \end{defn}

We now have all the definitions to set up our index. 

    \begin{defn}[The index of a locally-comparable pair of $\rho$-invariant pure states]\label{def:the index of a pair of pure states}

    Let $\rho_t$ be as above and $(\omega_1,\omega_2)$ be a pair of $\rho$-locally-comparable pure states. Assume moreover that both $\omega_1$ and $\omega_2$ are $\rho$-invariant, that is,
    \eq{
        \omega_j \circ \rho_t = \omega_j \qquad (t\in\RR,j=1,2)\,.
    }
    Then the index of the pair $(\omega_1,\omega_2)$ is defined as \eql{\label{eq:definition of general index}
        \calN_\rho(\omega_1,\omega_2) := \ii\omega_1\br{u^\ast \delta^\rho (u)} 
    } where $u\in\calD(\delta^\rho)$ is any unitary which obeys $\omega_2=\omega_1\circ\Ad{u}$.
    \end{defn}

    We note that in general there is no reason for the unitary $u$ above to be $\rho_t$-invariant, see \cref{Rem:SRE U}(iii) for an explicit example. Indeed, if this happens then $\delta^\rho (u) = 0$ and so $\calN_\rho(\omega_1,\omega_2)=0$. 

    The condition that the pair $(\omega_1,\omega_2)$ must be compatible with $\rho$ in the above sense, namely that they are $\rho$-locally-comparable and invariant, means they are generically \emph{not} path-connected, which is what the index measures.
    
    \begin{thm}[Properties of the index]\label{thm:PropertiesGeneralIndex}
        The index defined above has the following properties:
        \begin{enumerate}
            \item $\calN_\rho(\omega_1,\omega_2)$ does not depend on the choice of $u\in\calD(\delta^\rho)\cap\calU(\calA)$ such that $\omega_2=\omega_1\circ\Ad{u}$.
            \item If $\rho_{2\pi}=\mathrm{id}$, then  $\calN_\rho(\omega_1,\omega_2)\in\ZZ$.
            \item $\calN_\rho(\omega_1,\omega_2)$ has the following continuity property: \eql{
            \norm{\omega_1-\omega_2}<2\Longrightarrow \calN_\rho(\omega_1,\omega_2) = 0\,.
             }
           \item For any automorphism $\alpha$ of $\caA$ such that $\alpha\circ\delta^\rho = \delta^\rho\circ\alpha$, 
           \eql{\label{eq:automorpic invariance}
               \calN_\rho(\omega_1\circ\alpha,\omega_2\circ\alpha) = \calN_\rho(\omega_1,\omega_2).
           }
            \item $\calN_\rho\br{\omega,\omega} = 0$ for any $\rho$-invariant pure state $\omega\in\calP(\calA)$.
            \item If $\omega_1,\omega_2,\omega_3$ are three pairwise $\rho$-locally-comparable pure states, all of which are $\rho$-invariant, then \eql{\label{eq:additivity}
            \calN_\rho(\omega_1,\omega_2)+\calN_\rho(\omega_2,\omega_3) = \calN_\rho(\omega_1,\omega_3)\,.
            }
            \item The index is anti-symmetric \eql{
                \calN_\rho(\omega_1,\omega_2) = -\calN_\rho(\omega_2,\omega_1)\,.           
            }
            \item We have the following additivity with respect to tensor products. Let $\widetilde{\calA}$ be another C*-algebra and $\widetilde{\rho}$ as above; Let further $\widetilde{\omega_1},\widetilde{\omega_2}$ be two $\widetilde \rho$-locally-comparable pure states on $\widetilde{\calA}$ which are $\widetilde{\rho}_t$-invariant. Then
            \eql{
                \calN_{\rho\otimes\widetilde{\rho}}\br{\omega_1\otimes\widetilde{\omega_1},\omega_2\otimes\widetilde{\omega_2}} = \calN_\rho(\omega_1,\omega_2) + \calN_{\widetilde{\rho}}(\widetilde{\omega_1},\widetilde{\omega_2})\,.
            }
        \end{enumerate}
    \end{thm}

To clarify further, let us assume temporarily that $\calA$ is a uniformly hyperfinite (UHF) algebra. In this case, it is well known (see e.g. \cite[Prop 3.2.52]{bratteli2012operator}) that one may find a sequence $\Set{q_n=q_n^\ast}_n\subseteq\calA$ such that \eq{
            \delta^\rho = \slim_{n\to\infty}\ii[q_n,\cdot]\,.
        } Then it is clear that 
        \eql{\label{eq:abstract index for UHF algebras}
            \calN_\rho(\omega_1,\omega_2) = \lim_{n}\br{\omega_1(q_n)-\omega_2(q_n)}\,.
        } Since this expression does not involve the unitary $u$, the index is well-defined. The remainder of this subsection does not assume that $\calA$ is UHF.

\begin{rem}
    If $(\omega_1,\omega_2)$ are $\rho$-locally-comparable and if $\tilde\omega_2 = \omega_2\circ\mathrm{Ad}_v$ for a $v\in\caU(\caA)$ with $\delta^\rho(v)=0$, then $(\omega_1,\tilde \omega_2)$ are also $\rho$-locally-comparable and
\eql{\label{eq:unitary invariance on the second state}
    \caN_\rho(\omega_1,\tilde\omega_2 ) = \caN_\rho(\omega_1,\omega_2)
}
by additivity \cref{eq:additivity} and $\caN_\rho(\tilde\omega_2, \omega_2) = 0$.
\end{rem}
    \begin{proof}[Proof of \cref{thm:PropertiesGeneralIndex}]
        (i) If we have $\omega_2 = \omega_1\circ\Ad{u} = \omega_1\circ\Ad{v}$ for two a-priori distinct unitaries $u,v\in\calU(\calA)$ then this readily implies $\omega_1 = \omega_1 \circ \Ad{u^\ast v}$. Let us define $w := u^\ast v$ whence $\omega_1= \omega_1\circ\Ad{w}$. Since $v = uw $, we have \eq{
        v^\ast \delta^\rho(v) = (uw)^\ast \delta^\rho (uw)= w^\ast u^\ast \delta^\rho(u) w + w^\ast \delta^\rho(w)\,.
        } Applying $\omega_1$ to both sides of this and using the $w$-invariance of $\omega_1$, it remains to show that $\omega_1(w^\ast \delta^\rho(w)) =0$ to get the statement. Using \cref{lem:invariant pure state lemma} below we have that \eql{\label{eq:pure state phase}
    \abs{\omega_1(w)}=1\,.} 

    Defining $g:=w-\omega_1(w)\Id$, we calculate \eq{
    \omega_1(|g|^2) &= \omega_1\br{\br{w-\omega_1(w)\Id}^\ast \br{w-\omega_1(w)\Id}} \\
    &= \omega_1\br{|w|^2} - \abs{\omega_1(w)}^2 \\
    &= 0
    } where we have used that $\omega_1(\abs{w}^2)=\omega_1(w^\ast \Id w) = \omega_1(\Id) = 1$. Now by Cauchy-Schwarz,
    \eql{\label{eq:kernel of state is two-sided ideal}
    \abs{\omega_1(g^\ast a)}^2\leq\omega_1(|g|^2)\omega_1(|a|^2) = 0
    } for all $a\in\calA$ which shows that $\ker\omega_1$ is a two-sided *-ideal within $\calA$. Hence \eq{
    \omega_1(w^\ast \delta^\rho(w)) &=
    \omega_1(\br{g+\omega_1(w)\Id}^\ast \delta^\rho\br{g+\omega_1(w)\Id}) \\
    &= \omega_1(g^\ast \delta^\rho(g)) 
    \\&= 0
    } where we used that $\omega_1\circ\delta^\rho = 0$ since $\omega_1$ is invariant under $\rho_t$ in the second equality and \cref{eq:kernel of state is two-sided ideal} in the last equality.
    
        (ii) Let $(\calH_1,\pi_1,\Omega_1)$ be the GNS triplet associated with $\omega_1$. Let $Q_1$ be the self-adjoint operator on~$\calH_1$ such that $\ep{\iu t Q_1}$ implements $\rho_t$; see \cref{lem:generator of automorphism with integer spectrum} right below. By construction, $Q_1\Omega_1 = 0$. 
        
        We claim that $\Omega_2:=\pi_1(u)\Omega_1$ is also an eigenvector of $Q_1$, namely $Q_1\Omega_2=n\Omega_2$ for some $n\in\ZZ$ since $Q_1$ has integer spectrum. To see this we proceed as follows. Since $(\omega_1,\omega_2)$ is a locally-comparable pair, $\omega_2$ is a vector state in $\caH_1$, namely the vector $\Omega_2 \equiv \pi_1(u) \Omega_1$ is such that \eq{
        \ip{\Omega_2}{\pi_1(a)\Omega_2} = \omega_2(a)\qquad(a\in\calA)\,.
        } Moreover, $\omega_2\circ\rho_t=\omega_2$ implies, after going to the GNS representation and taking the derivative with respect to~$t$ at $0$, that \eq{
        \ip{\Omega_2}{[Q_1,A]\Omega_2} = 0\qquad(A\in\calB(\calH_1))\,,
        } since $\Omega_2$ is in the domain of $Q_1$. We may take $A$ to be the orthogonal projection onto $\szpan(\Psi)$ for some $\Psi\in\calH_1$ to conclude that \eql{\label{eq:reality equation}
        \ip{\Omega_2}{Q_1\Psi}\ip{\Psi}{\Omega_
        {2}} \in \RR\qquad(\Psi\in\calH_1)\,.
        } Now if $\Omega_2$ is \emph{not} an eigenvector of $Q_1$, then there is a non-zero $\Phi\in\szpan(\Omega_2)^\perp$ such that \eq{
        Q_1 \Omega_2 = \lambda \Omega_2 + \Phi
        } for some $\lambda\in\CC$. Actually $\lambda\in\RR$ because $Q_1$ is self-adjoint. Now if we invoke \cref{eq:reality equation} with $\Psi
        =\Omega_2 +\ii \Phi$ we obtain \eq{
        \RR \ni &\:\ip{\Omega_2}{Q_1\br{\Omega_2+\ii\Phi}}\ip{\Omega_2+\ii\Phi}{\Omega_2} = \ip{\lambda \Omega_2 + \Phi}{\Omega_2+\ii\Phi} \norm{\Omega_2}^2\\   &=\br{\lambda\norm{\Omega_2}^2+\ii\norm{\Phi}^2} \norm{\Omega_2}^2,
        } which is a contradiction unless $\Phi=0$. Hence $\Omega_2$ is an eigenvector of $Q_1$, which has integer spectrum, so \eql{\label{eq:charge deficiency in GNS}
        \calN_\rho(\omega_1,\omega_2) = \ip{\Omega_1}{Q_1\Omega_1} -\ip{\Omega_2}{Q_1\Omega_2}=-\ip{\Omega_2}{Q_1\Omega_2} \in \ZZ\,.
        }
        
        (iii) Using the same GNS notation as above, if $\calN_\rho(\omega_1,\omega_2)\neq0$ then $\Omega_2$ has a non-zero eigenvalue for $Q_1$. Since $Q_1$ is a self-adjoint operator and $\Omega_1$ is a zero-eigenvalue eigenvector for $Q_1$, this implies that $\ip{\Omega_2}{\Omega_1}=0$.

        On the other hand, since both states $\omega_1,\omega_2$ are represented as vector states on the same Hilbert space, we have $\norm{\omega_1-\omega_2}
        =\sup_{a\in\calA:\norm{a}\leq 1}\abs{\ip{\Omega_1}{\pi_1(a)\Omega_1} - \ip{\Omega_2}{\pi_1(a)\Omega_2}}$. Denoting by $P_j$ the orthogonal projection onto the span of $\Omega_{\omega_j}$ for $j=1,2$, we have
        \eq{
        \norm{\omega_1-\omega_2}
        =\sup_{a\in\calA:\norm{a}\leq 1}\abs{\Tr(\pi_1(a)(P_1-P_2))} = \norm{P_1-P_2}_1\,.
        }
        The second equality is by duality since $\pi_1$ is irreducible, as $\omega_1$ is pure. Since $P_j$ are one-dimensional projections, the trace norm is easily computable and we conclude that
        \begin{equation}
            \norm{\omega_1-\omega_2} = 2\sqrt{1-\abs{\ip{\Omega_2}{\Omega_1}}^2}.
        \end{equation}
        Hence,
        \eq{
        \abs{\calN_\rho(\omega_1,\omega_2)}>0 \Longrightarrow \norm{\omega_1-\omega_2}=2\,.
        } We note in passing that the reader may want to compare this with a related converse statement about unitary equivalence of the states when $\norm{\omega_1-\omega_2}<2$, see~\cite[Corollary 9, Remark 10]{GlimmKadison1959}.

(iv) Since $(\omega_2,\omega_1)$ are $\rho$-locally-comparable, there is $u\in\caU(\caA)\cap\caD(\delta^\rho)$ such that $\omega_2 = \omega_1\circ\Ad{u}$. Therefore,
\begin{equation}
    \omega_2\circ\alpha = \omega_1\circ\Ad{u}\circ\alpha = (\omega_1\circ\alpha)\circ \Ad{\alpha^{-1}(u)}
\end{equation}
so that $(\omega_2\circ\alpha,\omega_1\circ\alpha)$ are locally comparable. They are in fact $\rho$-locally-comparable since $\alpha$ commutes with $\delta^\rho$ and so $\alpha^{-1}(\caD(\delta^\rho))\subset \caD(\delta^\rho)$. The computation
\eq{ -\ii\calN_\rho(\omega_1\circ\alpha,\omega_2\circ\alpha)
    = (\omega_1\circ\alpha)\left(\alpha^{-1}(u)\str\delta^\rho(\alpha^{-1}(u))\right)
    = \omega_1 (u\str\delta^\rho(u)) = -\ii\calN_\rho(\omega_1,\omega_2)
}
yields the invariance.

       (v) By~(i), we can pick $u=\Id$, for which $\delta^\rho(u)=0$ and so the index vanishes.

        (vi) To get the additivity of the index, assume that $\omega_2 = \omega_1 \circ \Ad{u}$ and $\omega_3 = \omega_2 \circ \Ad{v}$. Hence $\omega_3 = \omega_1 \circ \operatorname{Ad}_{v u}$ so that
        \eq{
            \calN_\rho(\omega_1,\omega_3) &= \ii \omega_1\br{\br{vu}^\ast\delta^\rho\br{vu}} \\
            &= \ii \omega_1\br{u^\ast v^\ast\br{\br{\delta^\rho (v)}u+v\delta^\rho (u)}} \\
            &= \ii \omega_2(v^\ast \delta^\rho (v)) +\ii \omega_1(u^\ast \delta^\rho (u))\\
            &= \calN_\rho(\omega_2,\omega_3) + \calN_\rho(\omega_1,\omega_2)\,.
        }
        Finally, (vii) follows from~(vi) with $\omega_3 = \omega_1$ and (v).

        (viii) The result follows immediately since $\delta^{\rho\otimes\widetilde{\rho}} = \delta^\rho\otimes\mathrm{id}+\mathrm{id}\otimes\delta^{\widetilde{\rho}}$.
    \end{proof}

\begin{lem}\label{lem:invariant pure state lemma}
        Let $\omega\in\calP(\calA)$ and $u\in\calU(\calA)$ be such that $\omega=\omega\circ\Ad{u}$. Then $\omega(u)\in U(1)$.
    \end{lem}
    \begin{proof}
    Let $(\calH,\pi,\Omega)$ be the GNS representation of $\omega$. Then the invariance condition implies \eq{
    \omega(a) = \ip{\Omega}{\pi(a)\Omega} = \omega(u\str a u) = \ip{\Omega}{\pi(u)\str \pi(a) \pi(u)\Omega}\,.}
 We find that both $\Omega$ and $\pi(u)\Omega$ are cyclic vectors that generate the same state, where $\pi$ is irreducible as $\omega$ is pure. This implies that $\pi(u)\Omega = \ee^{\ii\theta}\Omega$ for some $\theta\in\RR$ and so $\omega(u) = \ip{\Omega}{\pi(u)\Omega} = \ee^{\ii\theta}$ as claimed.
    \end{proof}
    
     \begin{lem}\label{lem:generator of automorphism with integer spectrum}
        Let $\Set{\rho_t}_{t\in\RR}\subseteq\operatorname{Aut}(\calA)$ be a one-parameter group of strongly-continuous automorphisms such that $\rho_{2\pi}=\mathrm{id}$ and let $\omega\in\calP(\calA)$ be $\rho$-invariant. Then in the GNS representation $(\calH_{\omega},\pi_{\omega},\Omega_{\omega})$ of $\omega$, there exists a self-adjoint operator $Q_{\omega}:\calH_{\omega}\to\calH_{\omega}$ such that \eq{
        \pi_\omega\br{\rho_t \br{a} } = \ep{-\iu t Q_{\omega}} \pi_\omega\br{a} \ep{\iu t Q_{\omega}} \qquad(t\in\RR\,,\quad a\in\calA)\,.
        } One may choose $Q_{\omega}$ such that $Q_{\omega}\Omega_\omega = 0$, in which case $\sigma(Q_{\omega})\subseteq\ZZ$. 
    \end{lem}
    \begin{proof}
        For any $t\in\bbR$, the existence of a unitary implementer $U_{\omega,t}$ of $\rho_t$ is a standard consequence of the uniqueness of the GNS representation. Strong continuity of $t\mapsto\rho_t$ implies strong continuity of $t\mapsto U_{\omega,t}$, and the existence of a self-adjoint generator $Q_\omega$ is ensured by Stone's theorem. Since $\omega$ is pure, $\pi_{\omega}$ is irreducible and hence $Q_\omega$ is defined up to an additive constant. The invariance of $\omega$ yields that $\Omega_{\omega}$ is an eigenvector of $Q_\omega$, which can be fixed by imposing $Q_\omega\Omega_{\omega}=0$. With this, for any $a\in\caA$, $\ep{2\pi\iu Q_\omega}\pi_{\omega}(a)\Omega_{\omega} = \pi_{\omega}(\rho_{2\pi}(a))\Omega_{\omega} = \pi_{\omega}(a)\Omega_{\omega}$, which implies that $\ep{2\pi\iu Q_\omega} = \Id$ by cyclicity of GNS representation, so that $\sigma(Q_\omega)\subseteq\ZZ$.
    \end{proof}

    \begin{rem}
        The assumption of local-comparability and the standard \cref{lem:generator of automorphism with integer spectrum} immediately imply that the index is just the difference of eigenvalues of a charge operator which has integer spectrum, see \cref{eq:charge deficiency in GNS}, and hence the name `charge deficiency'. The interest of phrasing our result in the abstract algebraic setting is two-fold. First of all, as we shall see in \cref{sec:IQHE}, establishing the $\rho$-local-comparability of two states is the essence of the problem in applications. Secondly, it allows us to phrase the automorphic invariance \cref{eq:automorpic invariance} since the states $\omega_1$ and $\omega_1\circ\alpha$ are not in general unitarily equivalent and therefore not representable as vector states in each other's GNS representation. 
    \end{rem}

    Of course, the lemma also shows that the period $2\pi$ of the group $\rho_t$ is, at least mathematically, an arbitrary choice: If instead $\rho_{2\pi/\lambda}=\rm id$ for some $\lambda>0$ (but keeping $Q$ with integer spectrum) then $\widetilde{\calN_\rho}(\omega_1,\omega_2):=\lambda^{-1}\calN_\rho(\omega_1,\omega_2)\in\ZZ$.
    Finally, we emphasize again that the requirement that the intertwining unitary $u$ obey $u\in\calD(\delta^\rho)$ is not vacuous with the following example.
    \begin{example}
        Let $\calA=\operatorname{CAR}(\ell^2(\NN))$ with $\Set{e_j}_{j\geq1}$ the canonical position basis of $\ell^2(\NN)$ so $a_j\equiv a(e_j)$ are the annihilation operators associated with those basis elements. Then \eq{
        q_N := \sum_{j=1}^N a_j\str a_j
        } defines a sequence of bounded self-adjoint elements in $\calA$ which approximates a derivation $\delta$ as \eq{
        \delta(a) \equiv \lim_N \ii [q_N,a]\qquad(a\in\calD(\delta))
        } where $\calD(\delta)$ is the space of elements $a$ for which the limit exists. Define \eq{
        a := \sum_{j=1}^\infty j^{-3/2} a_1^\ast \cdots a_j^\ast
        } which exists since $\norm{a} \leq \sum_j j^{-3/2} < \infty$. We may then calculate \eq{
        [q_N, a_1^\ast \cdots a_j^\ast] = \min\br{\Set{N,j}} a_1^\ast \cdots a_j^\ast
        } and hence \eq{
        [q_N, a] = \sum_{j=1}^N j^{-1/2} a_1^\ast \cdots a_j^\ast + N \sum_{j=N}^\infty j^{-3/2} a_1^\ast \cdots a_j^\ast\,.
        } We see that the limit $N\to\infty$ of the above expression does not exist so that $a\notin\calD(\delta)$.

        Now take $h := a + a^\ast$ and $u := \br{\Id+\ii h}\br{\Id-\ii h}^{-1}$ to get a unitary element which is not in $\calD(\delta)$. As a result, for any $\omega\in\calP(\calA)$ we may define $\widetilde{\omega} := \omega\circ\Ad{u}$ such that $\omega,\widetilde\omega$ are locally comparable but not $\rho$-locally comparable.
    \end{example}

    \subsection{Connection with the index of a pair of projections on a Hilbert space \label{subsec:IPP}}
    Let $\calH$ be a separable Hilbert space and $P_1,P_2$ be two self-adjoint projections such that \eql{P_1-P_2\in\calK(\calH)\,.} 
    In this setting it is well-known that a topological index is associated with the pair $(P_1,P_2)$ (see e.g. \cite{AvronSeilerSimon} and earlier citations within) given by \eql{\label{eq:IPP}
    \operatorname{index}(P_1,P_2) \equiv \operatorname{index}_{\operatorname{im}(P_1)\to \operatorname{im}(P_2)}\br{\left.P_2P_1\right|_{\operatorname{im}(P_1)}} = \operatorname{index}_\calH\br{P_2P_1 + UP_1^\perp}
    } where the index on the right hand sides is the Fredholm index and $U$ is any unitary $U:\operatorname{im}(P_1^\perp)\to\operatorname{im}(P_2^\perp)$. In the special case that $P-Q\in \calJ_p(\caH)$, the ideal of $p$-Schatten class operators, for some $p$, we may also write \eq{
        \operatorname{index}(P_1,P_2) = \tr\br{\br{P_1-P_2}^{2p'+1}}
    } for any $p'\in\NN$ such that $2p'+1\geq p$. Finally, if $P_1-P_2\in \calJ_2(\caH)$ then by \cite[Theorem 3]{arveson2007diagonals}, we may write \eql{\label{eq:Arveson formula}
    \operatorname{index}(P_1,P_2) = \tr\br{P_2\br{P_1-P_2}P_2} + \tr\br{P_2^\perp \br{P_1-P_2}P_2^\perp } \,.
    }
    
    We turn to the relation of $\operatorname{index}(P_1,P_2)$ with the newly introduced many-body index $\caN_\rho(\omega_1,\omega_2)$.
    
    \begin{rem}\label{rem:BDF obstruction}
        Since any two projections $P_1,P_2$ on an infinite dimensional Hilbert space with infinite $\ker P_j,\,\operatorname{im} P_j$ admit a unitary $U$ which conjugates them, i.e., $P_2 = U^\ast P_1 U$, one may be tempted to conclude that if $P_1-P_2$ is ``small'' then so is $U-\Id$ (so as to have any hope to implement it in $\calA)$. However, it is well-known (see e.g. \cite[Theorem 1.2]{LoreauxNg2020} or \cite{BDF1973}) that a unitary $U\in\calU(\calH)$ can be found such that both $U-\Id\in\calK(\calH)$ and $P_2=U^\ast P_1 U$ iff \eq{\operatorname{index}(P_1,P_2)=0\,.} This indicates that to capture a non-trivial index we most likely need to find a unitary in the CAR algebra which is not the second quantization of a Hilbert space one.

        As we shall see below, there will be additional obstructions beyond the purely Hilbert space index obstruction.
    \end{rem}
    
    Let now $\calA=\operatorname{CAR}(\calH)$ (see \cref{sec:algebraic framework} for more details) be the algebra of observables corresponding to systems of many fermions. Any self-adjoint projection $P$ induces a pure state $\omega_{P}\in\calP(\calA)$, the so-called quasi-free state associated to $P$, given by the Gaussian formula \eql{\label{eq:quasifreefromP}
        \omega_{P}(a(f_1)^\ast \cdots a(f_m)^\ast a(g_n)\cdots a(g_1)) :=\delta_{n,m}\det\br{\Set{\ip{g_i}{P f_j}}_{i,j=1,\ldots,n}},
    } where $f_i,g_i\in\calH$, and extended linearly to all polynomials and by continuity to all of $\calA$.

     The canonical `charge' $U(1)$ *-automorphism $\rho_t$ is defined by\eql{\label{eq:U(1) on CAR}
        \rho_t(a(f)) = \ep{-\iu t} a(f),\qquad \rho_t(a\str (f)) = \ep{\iu t} a\str(f),
   }
   for any $f\in\calH$ and $t\in\bbR$. Note that $\rho_{2\pi} = \mathrm{id}$. For any orthogonal projection $P$, the state $\omega_P$ is automatically invariant under $\rho_t$.

   To contextualize the ensuing discussion, let us recall a few basic facts:
   \begin{enumerate}
       \item The Shale-Stinespring condition for quasi-free states on CAR($\calH$) \cite{ShaleStinespring}: $\omega_{P_1},\omega_{P_2}$ have unitarily equivalent GNS representations if and only if $P_1-P_2 \in \caJ_2(\caH)$.
       \item The Kadison transitivity theorem for general C*-algebras \cite[Corollary 8]{GlimmKadison1959}: If two pure states $\omega,\widetilde{\omega}\in\calP(\calA)$ have unitarily equivalent GNS representations then there exists some $u\in\calU(\calA)$ such that $\widetilde{\omega} = \omega \circ \Ad{u}$.
       \item The condition on Bogoliubov automorphisms of CAR$(\calH)$ to be inner \cite[Theorem 5]{araki1971quasifree}: $U-\Id\in \caJ_1(\caH)$ if and only if $\exists{\Gamma}(U)\in\calU(\calA)$ such that $a(Uf) = {\Gamma}(U)\str a(f){\Gamma}(U)$ for all $f\in\calH$.
   \end{enumerate}

    We are now ready to connect the index in \cref{eq:IPP} with the one defined in \cref{eq:definition of general index}. For this, given the basic facts stated above, and in particular~(iii), we need to strengthen the assumption on $P_1-P_2$ from compact to trace-class.
    \begin{thm}\label{thm:correspondence between many-body and single-particle indices}
        Let $P_1,P_2$ be self-adjoint projections such that \eq{P_1-P_2\in \mathcal{J}_1(\caH).} Then $\omega_{P_1}$ and $\omega_{P_2}$ are $\rho$-locally-comparable and \eq{
            \operatorname{index}(P_1,P_2) = \calN_{\rho}(\omega_{P_1},\omega_{P_2})\,.
        }
    \end{thm}
    This theorem identifies $\calN_\rho$ as the many-body analog of the notion of an index of a pair of projections, albeit with a certain mathematical gap: if $P_1-P_2\in\calK(\calH)$ but is not Hilbert-Schmidt, then by the above basic facts $\omega_{P_1}$ and $\omega_{P_2}$  cannot be locally-comparable: indeed, if they were, they would be GNS unitarily equivalent and hence violate the Shale-Stinespring condition. If $P_1-P_2$ is Hilbert-Schmidt but not trace-class, then by the above there would be a unitary conjugating them in the CAR algebra, but since that unitary is given abstractly by the Kadison transitivity theorem and is not the second quantization of any unitary, we do not know how to establish that that unitary necessarily lies in the domain of $\delta^\rho$.
    
    Ultimately, a generalization of the many-body index may exist, which relies on a weaker condition than that $(\omega_1,\omega_2)$ be locally comparable, and which would always be defined for quasi-free states $(\omega_{P_1},\omega_{P_2})$ as soon as $P_1-P_2\in\calK$. This problem is not entirely academic since in the interesting case of the integer quantum Hall effect, $P-L^\ast P L\in \calJ_3\setminus \calJ_2$ where \eql{\label{eq:Laughlin unitary}L\equiv\ep{\ii\arg\br{X_1+\ii X_2}}} is the Laughlin unitary \cite{AvronSeilerSimon_Charge} implementing one quantum of a magnetic flux insertion. See also~\cref{example:compact but not Hilbert-Schmidt} below for a simple example of this type.

     \begin{example}[Non-zero index yet unitary exists thanks to interactions] \label{example:shift} The obstruction outlined in \cref{rem:BDF obstruction} does \emph{not} imply that we cannot implement non-zero indices. Indeed, pick $\calH := \ell^2(\ZZ)$ and let $P$ be the multiplication operator by the indicator function $\chi_\NN$, namely the projection onto the RHS of space. Let $R$ be the bilateral right shift operator $R:\delta_x\mapsto\delta_{x+1}$. Then $P_R=R^\ast P R$ is the projection given by $\chi_{\NN\cup\Set{0}}$ and $P_R-P$ is the finite rank projection given by $\chi_{\Set{0}}$, and $\operatorname{index}(P,P_R) = -1$. Of course this means $P-P_R \in \mathcal{J}_2(\caH)$, equivalently, $[R,P]\in \mathcal{J}_2(\caH)$. 
     
     The operator
\eq{
     u := a_0 + a_0^\ast\,.
     } 
     is unitary since $u^\ast u = u^2 = (a_0 + a_0^\ast)^2 = \Id$, such that 
     \begin{equation*}
         u\str a_0 u = a_0^\ast\,,\quad\text{while}\quad u\str a_m u = - a_m \quad\text{whenever } m\neq0.
     \end{equation*}
    We claim that
      \eq{
     \omega_{P_R} = \omega_P \circ \Ad{u}\,.
     }

     Indeed, if $n,m\neq 0$,
     \eq{
     \omega_P(u\str a_m^\ast a_n u) = \omega_P( a_m^\ast a_n ) = \langle \delta_n,P\delta_m\rangle = \delta_{m,n}\chi_\NN(n) = \omega_{P_R}(a_m^\ast a_n)
     } and if $n=m=0$ then 
     \eq{
     \omega_P(u^\ast a_0^\ast a_0 u) = \omega_P(a_0 a_0^\ast) = \omega_P(\Id - a_0^\ast a_0)=1 = \omega_{P_R}(a_0^\ast a_0)\,,} while both $\omega_P(u\str a_m^\ast a_n u)$ and $\omega_{P_R}(a\str_m a_n)$ vanish if $n\neq m$. Since
     \begin{equation*}
         \delta^\rho(a_0) = \frac{d}{dt}\left.\rho_t(a_0)\right\vert_{t=0} = -\iu a_0
     \end{equation*}
     we see that $u\in\caD(\delta^\rho)$ and $-\iu \delta^\rho(u) = -a_0 + a_0\str$, so that
     \eq{
     \calN_\rho(\omega_P,\omega_{P_R}) = \ii\omega_P(u\str \delta^\rho(u)) =  -\omega_P(-a_0\str a_0 + a_0a_0\str) = -\omega_P(\Id - 2 a_0\str a_0) 
     = -1\,.
     }
    \end{example}

     \begin{example}[Compact but not Hilbert-Schmidt difference of projections so no unitary exists]\label{example:compact but not Hilbert-Schmidt}
        Let $\calH := \ell^2(\NN)\otimes\CC^2$. Let $P$ be given by \eq{
        P_{nm} := \delta_{nm} \begin{bmatrix}
            1 & 0 \\ 0 & 0
        \end{bmatrix}\,,
        } i.e., projection onto the top of each dimer. Let $U$ be the unitary given by \eq{
        U_{nm} := \delta_{nm} \begin{bmatrix}
            \cos(\theta_n) & \sin(\theta_n) \\ -\sin(\theta_n) & \cos(\theta_n)
        \end{bmatrix} 
         } where $\Set{\theta_n}_n$ is some sequence of angles to be determined below. If $Q := U^\ast P U$, then \eq{
P_{nm}-Q_{nm}=\delta_{nm}\begin{bmatrix}
  \sin^{2}\theta_n & -\cos\theta_n\sin\theta_n\\
  -\cos\theta_n\sin\theta_n & -\sin^{2}\theta_n
\end{bmatrix}\,.
         } Since the singular values of $P-Q$ are $\Set{\abs{\sin(\theta_n)}}_{n\geq1}$ it is clear that if we pick $\theta_n := \arcsin(x_n)$ for some sequence $\set{x_n}_n\subseteq(0,1)$ we can engineer whichever summability we want. For example if we take $x_n = n^{-\beta}$ for any $\beta\in(1/3,1/2]$ then we guarantee that $\Set{x_n^2}_n$ is not summable but $\Set{x_n^3}_n$ is. Note that as soon as $x_n\to0$ we have $P-Q\in\calK(\calH)$.

         Now, for $\Set{x_n}_n$ in $\ell^3$ but not $\ell^2$ there does not exist a unitary $u\in\calU(\calA)$ such that $\omega_Q = \omega_P\circ\Ad{u}$. Indeed, the existence of such a unitary would violate the Shale-Stinespring condition \cite{ShaleStinespring}.
    \end{example}

Before presenting the proof of \cref{thm:correspondence between many-body and single-particle indices}, let us make a final remark. Suppose we knew that $\omega_1$ and $\omega_2$ (with $\omega_j\equiv\omega_{P_j}$) are $\rho$-locally-comparable. Then the equivalence of the two indices would be immediate. Indeed, starting from \cref{eq:abstract index for UHF algebras} and approximating $\delta^\rho = \lim_N\ii[q_N,\cdot]$ with  $q_N := \sum_{j=1}^N a\str(e_j)a(e_j)$ where $\Set{e_j}_j$ is any ONB of $\calH$, we obtain using~\cref{eq:quasifreefromP} 
        \eq{
        \calN_\rho(\omega_1,\omega_2) &= \lim_N \br{\omega_1(q_N) - \omega_2(q_N)} = \lim_N \sum_{j=1}^N\br{\ip{e_j}{P_1 e_j}-\ip{e_j}{P_2e_j}} \\
        &= \lim_N \tr\br{\br{P_1-P_2}\chi_{\Set{1,\cdots,N}}(X)}
        } where $X$ is the position operator w.r.t. the basis $\Set{e_j}_j$, i.e. $X e_j \equiv j e_j$ for any $j\in\NN$. Since we have $\chi_{\Set{1,\cdots,N}}(X)\to\Id$ strongly as $N\to\infty$, and since $P_1-P_2\in \mathcal J_1$ we are finished, since the product of a strongly convergent sequence with a trace-class operator is a sequence converging in trace-class norm, i.e., we find
        \eq{
        \calN_\rho(\omega_1,\omega_2) &= \tr\br{P_1-P_2} = \operatorname{index}(P_1,P_2)\,.
        } 
We also note that in the case $P_1-P_2\in \caJ_2(\caH)$ the proof given (using a particular choice of the ONB $\Set{e_j}_j$) goes through by invoking the explicit formula~\cref{eq:Arveson formula}. In the proof below, we rely on the trace-class assumption to construct the unitary intertwiner between $\omega_1$ and $\omega_2$, which lies in $\calD(\delta^\rho)$.
        So the main difficulty is to establish the $\rho$-local-comparability of $\omega_1,\omega_2$, i.e., to find the unitary which conjugates them and show it lies in $\calD(\delta^\rho)$. 

        The proof of~\cref{thm:correspondence between many-body and single-particle indices} relies on two steps.
        \begin{enumerate}
            \item We first extract the `excess states' out of the difference $P_1-P_2$ which are responsible for a non-zero index, see~\cref{lem:P-Q_decomposition} below. Because of the obstruction in \cref{rem:BDF obstruction}, the unitary which implements this unfolding is necessarily not the second quantization of a Hilbert space one, but it can be constructed explicitly. In particular, it is in $\calD(\delta^\rho)$ since the charge deficiency is finite. Indeed, it is very much in the spirit of \cref{example:shift}.
            \item As a result, we obtain two projections whose index is zero and whose difference is trace-class. It remains to use the second quantization provided by~\cref{lem:Quasi-free implementation}, which does not change the charge deficiency. Since all unitaries are explicit, the equality of indices follows from a computation.
        \end{enumerate}

\begin{proof}[Proof of \cref{thm:correspondence between many-body and single-particle indices}]

Using \cref{lem:P-Q_decomposition} below, since $P_1-P_2 \in \calJ_1(\calH)$, we write \eq{P_2 = \widetilde P + N_+ - N_-\,,}
with $\widetilde P=VP_1V\str$ and $V-\Id\in \calJ_1(\calH)$. By \cref{lem:Quasi-free implementation} below, $\omega_{\widetilde P}$ and $\omega_{P_1}$ are $\rho$-locally comparable with $\Gamma(V\str)\in \mathcal U(\calA) \cap \calD(\delta^\rho)$ and
    \begin{equation}\label{eq: Bog on states 2}
        \omega_{\widetilde P} = \omega_{P_1}\circ\Ad{\Gamma(V\str)}\,,
    \end{equation}
since $\omega_{\widetilde P}(a\str(f) a(g)) = \langle g, VP_1V\str f\rangle = \omega_{P_1}(a\str(V\str f) a(V\str g))$.
Let $n_\pm=\mathrm{rank}(N_\pm)$ and $\{f_i\}$ and $\{g_i\}$ be orthonormal bases of the range of $N_-$ and $N_+$ respectively. We consider
$$
v_+ = \prod_{i=1}^{n_+} (a(g_i)+a^*(g_i)), \qquad v_-=\prod_{i=1}^{n_-} (a(f_i)+a^*(f_i)).
$$
One checks that $v_\pm \in \calU(\cA)$ as products of unitary factors, see again~\cref{example:shift}. Moreover, $v_\pm \in \calD(\delta^\rho)$ as polynomials in creation and annihilation operators. We claim that 
$$
u := v_+v_- \Gamma(V)\str \in \mathcal U(\calA) \cap \calD(\delta^\rho)
$$
is such that
\begin{equation}\label{eq:P2Ptildeconugacy}
\omega_{P_2} = \omega_{\widetilde P} \circ \Ad {v_+v_-} = \omega_{P_1}\circ \Ad u.
\end{equation}
We prove the first equality in \eqref{eq:P2Ptildeconugacy}, the second one is just \cref{eq: Bog on states 2}. Since
\begin{equation*}
    (a\str(\xi) + a(\xi)) a(\phi)  = \langle \phi,\xi\rangle\Id - a(\phi)\big( a\str(\xi) + a(\xi)\big)
\end{equation*}
    we see that for $v_\xi = a\str(\xi) + a(\xi)$,
    \begin{equation*}
    \omega_{\widetilde P}\circ \Ad{v_\xi}(a\str (\psi)a(\phi))
    = \langle\phi,(\widetilde P + P_\xi - (\widetilde P P_\xi + P_\xi\widetilde P))\psi\rangle
\end{equation*}
where $P_\xi = \vert \xi\rangle\langle \xi\vert$. In particular,
\begin{equation*}
    \begin{cases}
        \omega_{\widetilde P}\circ \Ad{v_\xi}(a\str (\psi)a(\phi)) = \langle\phi,(\widetilde P - P_\xi )\psi\rangle &\text{if}\quad \widetilde P P_\xi = P_\xi,
    \\
    \omega_{\widetilde P}\circ \Ad{v_\xi}(a\str (\psi)a(\phi))
    = \langle\phi,(\widetilde P + P_\xi )\psi\rangle & \text{if}\quad \widetilde P P_\xi = 0.
    \end{cases}
\end{equation*}

Applying this recursively with $\xi$ running through the basis of $N_+$ and then the basis of $N_-$ and recalling that $\widetilde P N_- = N_-$ and $\widetilde P N_+ = 0$ yields
\begin{equation*}
    \omega_{\widetilde P}\circ \Ad{v_+v_-}(a\str (\psi)a(\phi))
    = \langle\phi,(\widetilde P + N_+ - N_-)\psi\rangle
    =\omega_{P_2}(a\str (\psi)a(\phi))
\end{equation*}
which proves \cref{eq:P2Ptildeconugacy}. Thus $\omega_{P_2}$ and $\omega_{P_1}$ are $\rho$-locally comparable.

For completeness let us show the equivalence of the indices using the intrinsic decomposition just exhibited. We compute $\caN_{\rho}(\omega_{P_1},\omega_{P_2})=\iu\omega_{P_1}(u\str\delta^\rho(u))$. Since $\delta^\rho(\Gamma(V))=0$, we focus on $v_-\str v_+\str \delta^\rho(v_+)v_- + v_-\str v_+\str v_+\delta^\rho(v_-)$. Since
\begin{equation*}
    \delta^\rho(a\str(f) + a(f)) = \iu(a\str(f) - a(f))
\end{equation*}
for any normalized $f$, and therefore
\begin{equation*}
  (a\str(f) + a(f))  \delta^\rho(a\str(f) + a(f)) = \iu(\Id -2a\str(f)a(f))
\end{equation*}
we conclude that $v_+\str\delta^\rho(v_+) = \iu(n_+ - 2Q_+)$ and so
\begin{equation*}
   \iu v_-\str v_+\str \delta(v_+)v_- + \iu v_-\str v_+\str v_+\delta^\rho(v_-) = -(n_- - 2 Q_-) - (n_+ - 2 Q_+).
\end{equation*}
Here $Q_- = \sum_{i=1}^{n_-}a\str(f_i) a(f_i)$ and $Q_+ = \sum_{i=1}^{n_+}a\str(g_i) a(g_i)$. Hence
\begin{align*}
    \caN_\rho(\omega_{P_1},\omega_{P_2})
    &= - \omega_{P_1}\left(\Gamma(V)((n_- - 2 Q_-) + (n_+ - 2 Q_+))\Gamma(V\str)\right) \\
    &=-n_- - n_+ + 2\omega_{\widetilde P}\left(Q_-+  Q_+\right)
    = n_- - n_+
\end{align*}
since $\omega_{\widetilde P}(Q_-) = n_-$ while $\omega_{\widetilde P}(Q_+)=0$ by~(\ref{eq:quasifreefromP}).
    
\end{proof}

\begin{lem}\label{lem:P-Q_decomposition}
If $P_1,P_2$ are two orthogonal projections on a Hilbert space $\calH$ such that $P_1-P_2\in \mathcal J_p(\calH)$ for $p\in \mathbb N^*$, there exist finite rank orthogonal projections $N_+,N_-$ with $N_+N_-=0$ and a unitary $V$ with $V-\Id \in \mathcal J_p(\calH)$ such that
\begin{equation}\label{eq:P2P1_wolddecomposition}
        P_2 = V P_1 V^* + N_+ - N_-
\end{equation}
Moreover $\mathrm{index}(P_1,P_2) = \mathrm{tr}(N_--N_+)$, and for $\tilde P=VP_1V^*$ one has $\widetilde P N_+ = 0$ and $\widetilde P N_-=N_-$. 
\end{lem}
The proof of this lemma can be found in \cite[Proof of Prop. 5.2]{BachmannBolsRahnama24}, see also \cite[Sec~4.2.1]{bols2024absolutely} or \cite[Sec~VII.B]{cedzich2018topological}.

        \begin{lem}\label{lem:Quasi-free implementation}
            If $V$ is a unitary operator on $\caH$ such that $V-\Id\in \caJ_1(\caH)$, the *-automorphism defined by $a(f)\mapsto a(V f)$ is unitarily implementable by $\Gamma(V)\in\caU(\caA)$, namely
            \begin{equation*}
                a(V f) = \Gamma(V)\str a(f)\Gamma(V)
            \end{equation*}
            for all $f\in\caH$. Moreover, $\delta^\rho(\Gamma(V)) = 0$.
        \end{lem} 
This is a known result up to possibly the last property, see e.g.~\cite[Theorem~5]{araki1971quasifree}. For the convenience of the reader we include its proof here.
\begin{proof}[Proof of \cref{lem:Quasi-free implementation}]
We recall the explicit construction of the implementer $\Gamma(V)\in\caA$ given in~\cite{Fredholm} whenever $V-\Id$ is of finite rank. For rank one operators $A_i=|f_i\rangle\langle g_i|$, $(i=1,\ldots,n)$,
we set
\begin{equation}
\label{eq:DefDGamma}
\dd\Gamma(A_1,\,\ldots,\,A_n)=
a\str(f_n)\cdots a\str(f_1)\,a(g_1)\cdots a(g_n)\,.
\end{equation}
The definition is extended by multilinearity to operators $A_i$ of finite
rank. The result is independent of the particular decomposition into rank
one operators. Then if $V-\Id$ has finite rank, we set
\begin{equation}
\label{eq:DefGamma}
\Gamma(V)=\Id + \sum_{n=1}^\infty\frac{1}{n!}\dd\Gamma(V-\Id,\,\ldots,\, V-\Id )\,.
\end{equation}
The sum is finite, because the terms with $n>\mathrm{rank}(V-\Id)$ vanish by the CAR. 

A calculation, see~\cite{Fredholm}, then shows that $\Gamma(V)$ is indeed an implementer
\begin{equation}
\Gamma(V)a\str(f) =a\str(Vf)\Gamma(V)
\end{equation}
and that $\Gamma(V_1)\Gamma(V_2) = \Gamma(V_1V_2)$. In particular, $\Gamma(V)$ is unitary if $V$ is. Moreover, the definition~\cref{eq:DefDGamma} immediately implies that $\rho_t(\dd\Gamma(A_1,\,\ldots,\,A_n)) = \dd\Gamma(A_1,\,\ldots,\,A_n)$, where $\rho_t$ is the $U(1)$-automorphism defined in~\cref{eq:U(1) on CAR}. This yields in turn that $\rho_t(\Gamma(V)) = \Gamma(V)$ by the definition \cref{eq:DefGamma}. Hence $\delta^\rho(\Gamma(V))=0$.

It remains to extend the map $\Gamma$ to unitaries such that $V-\Id$ is trace class. For rank 1 operators, \cref{eq:DefDGamma} and the CAR yield
\begin{equation*}
    \Vert \dd\Gamma(A_1,\,\ldots,\,A_n)\Vert \leq \prod_{i=1}^n\Vert f_i\Vert\Vert g_i\Vert
\end{equation*}
and so if $A_i$ is a rank $r_i$ operator, we have that
\begin{equation*}
    \Vert \dd\Gamma(A_1,\,\ldots,\,A_n)\Vert 
    \leq \inf_{\{f\},\{g\}}\sum_{i_1=1}^{r_1} \Vert f_{i_1}\Vert\Vert g_{i_1}\Vert \cdots \sum_{i_n=1}^{r_n} \Vert f_{i_n}\Vert\Vert g_{i_n}\Vert = \prod_{i=1}^n \Vert A_i\Vert_1
\end{equation*}
where the infimum is over all decompositions of the operators $A_j = \sum_{i_j=1}^{r_j}\vert f_{i_j}\rangle\langle g_{i_j}\vert$. Hence $\dd\Gamma$ is a bounded linear transformation from finite rank operators equipped with the trace norm to the algebra, and it extends to a bounded linear transformation $\dd\Gamma:\caJ_1(\caH)^{\otimes n}\to\caA$, with $\Vert\dd\Gamma\Vert = 1$ for all $n\in\mathbb{N}$. This now implies that the series in~\cref{eq:DefGamma} is convergent since
\begin{equation*}
    \left \Vert \frac{1}{n!}\dd\Gamma(V-\Id,\,\ldots,\, V-\Id )\right\Vert
    \leq \frac{1}{n!} \Vert V-\Id\Vert_1^n.
\end{equation*}
In fact, we have that $\Vert \Gamma(V)\Vert \leq \ep{\Vert V-\Id\Vert_1}$.
\end{proof}

\section{Integer Quantum Hall Effect \label{sec:IQHE}}

We now turn to a concrete setting where the abstract index defined and analyzed in \cref{sec:abstract_index} is realized as the charge deficiency associated with the adiabatic pumping of a unit flux through a 2-dimensional electron gas. In the present interacting setting this is closely related to~\cite{kapustin2020hall}. The key difficulty is to find an appropriate class of states on which flux insertion can be realized by a unitary element of the algebra. A construction of the unitary was however implemented in~\cite{BachmannBolsRahnama24} upon which we shall rely in this section.

\subsection{Construction overview and comparison with previous work}
\label{sec:summary}

\begin{figure}[!ht]
    \centering 
    \tikzset{
	box/.style={rectangle, rounded corners, draw=blue!60, fill=blue!10, thick, minimum width=2.5cm, align=center},
	arrow/.style={-Latex, thick}
    }
    	\begin{tikzpicture}[node distance=1cm and 1cm]
		\node[box] (start) {State $\omega$ \\$Q$-sym. and SRE};
		\node[box, right=of start] (parent) {Parent Hamiltonian $H$ \\ {\hfill \scriptsize \cref{lem:Q parent}} };
		\node[box, below=of parent] (HLgauge) {Half-plane gauge transf. \\$H^\up(\phi) = \rho^\up_{-\phi} (H)$};
		\node[box, right=of HLgauge] (HLstate) {$\omega^\up_\phi=\omega \circ \rho_\phi^\up$};
		\node[box, below=of HLgauge] (QAflow) {Quasi-adiabatic flow \\ $K(\phi)$ \\ {\scriptsize \hfill \cref{prop:defK}}};
		\node[box, right=of QAflow] (QAstate) {$\omega^\up_\phi=\omega \circ\alpha^K_\phi$};
		\node[box, below=of QAflow,align=left] (Defect) {Flux-insertion autom. \\ $D(\phi) = K(\phi)\mathds 1_{\{x_1<0\}}$ \\ $\gamma_\phi = \alpha^D_\phi$ \\ {\scriptsize\hfill \cref{def:flux_insertion_autom}}};
		\node[box, right=of Defect] (DefectState) {Defect state \\ $\omega^D=\omega \circ \gamma_{2\pi}$};
		\node[box, below=of DefectState] (existsU) {$\exists u \in \mathcal D(\delta^Q)\cap \mathcal U(\mathcal A)$ \\ $\omega^D=\omega \circ \mathrm{Ad}(u)$ \\ {\hfill \scriptsize \cref{prop:existence_of_U}}};
        \node[box, below=of existsU] (index) {$i(\omega,Q) = \mathcal N(\omega,\omega^D)$ \\ {\hfill \scriptsize \cref{prop:SRE_index}}};
		
		\draw[arrow] (start) -- (parent);
		\draw[arrow] (parent) -- (HLgauge);
		\draw[arrow] (HLgauge) -- (HLstate);
		\draw[arrow] (HLgauge) -- (QAflow);
		\draw[arrow] (QAflow) -- (QAstate);
		\draw[arrow] (QAflow) -- (Defect);
		\draw[arrow] (Defect) -- (DefectState);
		\draw[arrow] (DefectState) -- (existsU);
        \draw[arrow] (existsU) -- (index);

        \draw[] (9,-3.3) node {\rotatebox{90}{$=$}};
        
		\begin{scope}[yshift=-0.3cm]
    		\fill[red!20]  (1.2,-1.9) rectangle (2.2,-1.4);
    		\draw[->] (1.2,-1.9) -- (2.2,-1.9);
    		\draw[->] (1.7, -2.3) -- (1.7,-1.4);
        \end{scope}
		
		\begin{scope}[yshift=-2.5cm]
			\fill[red!20]  (1.2,-2) rectangle (2.2,-1.8);
			\draw[->] (1.2,-1.9) -- (2.2,-1.9);
			\draw[->] (1.7, -2.3) -- (1.7,-1.4);
		\end{scope}
	
		\begin{scope}[yshift=-5cm]
			\fill[red!20,rounded corners]  (1.2,-2) rectangle (1.8,-1.8);
			\draw[->] (1.2,-1.9) -- (2.2,-1.9);
			\draw[->] (1.7, -2.3) -- (1.7,-1.4);
		\end{scope}

		\begin{scope}[yshift=-5cm,xshift=8.7cm]
				\fill[red!20]  (1.7,-1.9) circle (0.1);
				\draw[->] (1.2,-1.9) -- (2.2,-1.9);
				\draw[->] (1.7, -2.3) -- (1.7,-1.4);
		\end{scope}

		\draw[dashed,rounded corners] (1,-0.9) rectangle (11.4,-8);

        \draw (10.25,-0.6) node {Flux insertion};
	\end{tikzpicture}

    \caption{Overview of the construction from the SRE state $\omega$ to the defect state $\omega^D$ via magnetic flux insertion. The $0$-chains are progressively localized near the half-line. At $\phi=2\pi$, the two states only differ near the origin, which allows for the existence of a unitary $u$ so that they are $\rho^Q$-locally comparable. }
    \label{fig:roadmap}
\end{figure}
    
The construction being rather long, we summarize its main steps. First, we set in \cref{sec:algebraic framework} the algebraic framework by defining concretely the CAR algebra $\mathcal A$ and its local structure, $0$-chains and their associated dynamics (locally generated automorphism  or LGA), and the charge operator $Q$ associated to $U(1)$-gauge transformations $\rho$. This allows us to define $Q$-symmetric and short-ranged-entangled (SRE) states, which we aim to classify up to homotopy. The SRE property encodes the locality property of a state $\omega$. In particular, the existence of a parent Hamiltonian $H$, for which $\omega$ is a gapped ground state, is a direct consequence of the SRE property. Moreover, in contrast to single particle systems where electric charge is the identity operator, its many-body counterpart $Q$ is not always preserved, so we focus on the class of $Q$-symmetric states to study the integer quantum Hall effect.

Then, a central step is the construction of the flux insertion automorphism from \cref{sec:flux_insertion}, which starts from an SRE state $\omega$ and ends with the defect state $\omega^D$, where a unit of magnetic flux has been inserted at the origin; see \cref{fig:roadmap}.
Flux insertion relies on gauge transformation, quasi-adiabatic flow and truncation. The main goal is to localize the operations to a half-line, and mimics the half-line gauge flux insertion from single particle picture \cite{DeNittisSchulzBaldes2016spectral}. At one quantum of flux $\phi=2\pi$, the two states $\omega$ and $\omega^D$ only differ near the origin, allowing the existence of a unitary $u\in\calA$ relating the two. Namely, they are $\rho$-locally comparable and the formalism of \cref{def:the index of a pair of pure states} applies for $\mathcal N_\rho(\omega,\omega^D)$. This construction may look convoluted at first, but it is to our current knowledge the only way to construct a defect state which is $\rho$-locally comparable to $\omega$ (see also the discussion at the beginning of \cref{sec:flux_insertion}). An analogous version is presented in \cite{BachmannBolsRahnama24}, on which we rely.

\begin{figure}[ht]
    \centering 
    \tikzset{
	box/.style={rectangle, rounded corners, draw=blue!60, fill=blue!10, thick, minimum width=2.5cm, align=center},
	arrow/.style={-Latex, thick}
    }
    	\begin{tikzpicture}[node distance=1cm and 1cm]

        \node[box] (start) {$\omega$ invertible: \\ 
        $\exists \omega', \, \omega \hat\otimes\omega'$ SRE \\ and $Q$, $Q'$ symmetric};

        \node[box, right=of start] (flux) {Flux insertion with $Q+Q'$\\
        $\hat \omega^D$ on $\cA\hat \otimes \cA'$};

        \node[box, right=of flux] (index) {$\hat \omega^D$ is $Q$-symmetric \\
        $\mathcal I(\omega) = \mathcal N_{\rho^Q\hat \otimes \mathrm{id}}(\hat\omega,\hat\omega^D)$};
		
		\draw[arrow] (start) -- (flux);
		\draw[arrow] (flux) -- (index);
	\end{tikzpicture}

    \caption{Extending the construction to invertible states.}
    \label{fig:roadmap2}
\end{figure}

However, it is actually not possible to provide a quantum Hall model for which this index on SRE states is non-zero. Indeed, for a translation-invariant system of non-interacting electrons, the (stable) SRE condition is equivalent to the triviality of the bundle, hence the vanishing of the Hall conductance and accordingly the vanishing of our index. To model non-trivial Hall states, in \cref{sec:MBI} we extend the construction to invertible states. Those are states $\omega$ for which there exists some auxiliary state $\omega'$ over $\cA'$ such that $\omega \hat\otimes\omega'$ is itself SRE on $\cA\hat\otimes\cA'$. $\cA'$ is another CAR algebra which implements extra degrees of freedom on each lattice site, and $\hat \otimes$ is the stacking operation. In the non-interacting case, this amounts to considering the direct sum of two bundles having opposite Chern numbers; the Chern number is additive under direct sums. We then reproduce the previous construction, with the main subtlety of having two charge operators $Q$ and $Q'$ for each layer. The flux insertion is performed on the two layers with $Q + Q'$, but the symmetry appearing in our index is $Q$, acting on the first layer only. Physically, this amounts to a bona fide flux insertion across both layers, which is essential to have $\rho$-locally comparable states at the end of the process, but a measurement of charge transport only on the first, physical one. The construction is summarized in \cref{fig:roadmap2}. It leads to \cref{def: THE Index} and \cref{thm:main}, the main result of \cref{sec:MBI}.

We consider here exclusively the `bulk picture' of the quantum Hall effect. As pointed out in the introduction, the technicalities of the construction we provide in this section are heavily inspired by~\cite{BachmannBolsRahnama24}. There, the additional assumption of time reversal invariance implies that there is no net charge transport and the interest is on the remaining $\mathbb{Z}_2$-valued index which can be non-trivial even for SRE states. In the infinite planar geometry we consider here, the construction of a flux insertion automorphism for invertible states goes back to~\cite{kapustin2020hall}, which in turn was inspired by the flux threading procedure of~\cite{bachmann2020many,bachmann2018quantization} on large but finite tori. From the point of view of these works, the goal of this section is merely to point out that they provide a concrete physical setting in which the general abstract index of the previous section is realized. Still in the many-body, interacting setting, the Hall conductance and its quantization have also been described as a Thouless pump and~\cite{bachmann2022classification, KapustinArtymowicz} provide a different, albeit related picture of its quantization. For weakly interacting fermions, \cite{giuliani17} provides a rigorous renormalization group approach to the quantization of the quantum Hall conductance: while the previous works rely on the assumption of a spectral gap, this result proves the stability of the gap in the perturbative regime. Recently, \cite{teufel25} have extended the strongly interacting framework to some cases where the gap closes. This is in fact closely related to~\cite{hastings2015quantization}, which initiated the present line of work. While these works are all concerned with microscopic lattice systems, the quantum Hall effect was analyzed earlier from a field theoretic point of view, see e.g.~\cite{frohlich1991large} or the reviews~\cite{Froehlich, frohlich1997classification}. Finally, \cref{sec:free} will establish the equality of the Hall index defined in this work for invertible states with the Hall index for a pair of projections, and through it with the many indices long known to be equivalent to it, whenever they can be compared: spectral flow, Chern number or K-theoretic indices. In particular, we will prove that the trace-class assumption of the general Theorem~\ref{thm:correspondence between many-body and single-particle indices} follows from the locality of the Hamiltonian and the gap assumption.

\subsection{Algebraic framework}\label{sec:algebraic framework}

Fermionic observables are elements of the Canonical Anti-commutation Relation (CAR) algebra $\cA$ over the one-particle Hilbert space ${\caH} = \ell^2(\mathbb Z^2)\otimes \mathbb C^n$ where $n \in \mathbb N$ takes into account internal degrees of freedom such as spin. The algebra $\cA$ is the unital $C^*$-algebra generated by an identity  element $\Id$%, $\{a(f):f\in{h}\}$ 
and the annihilation operators $\{a(f):f\in{\caH}\}$, which satisfy the CAR
\begin{equation}
    \{a(f),a\str(g)\} = \langle g,f\rangle \Id,\qquad
    \{a(f),a(g)\} =\{a\str(f),a\str(g)\} = 0
\end{equation}
for all $f,g\in{\caH}$, and where $a\str(f) = (a(f))\str$ are the creation operators. By picking the orthonormal basis $\{\delta_x\otimes e_i:x\in\bbZ^2,i\in\{1,\ldots,n\}\}$ of ${\caH}$, the algebra is also generated by $a_{x,i},a\str_{y,j}$ for $x,y \in \mathbb Z^2$ and $i,j \in \{1,\ldots, n\}$, where $a_{x,i}^\sharp = a^\sharp(\delta_x\otimes e_i)$. Here and in the following $a^\sharp$ stands for either $a$ or $a\str$.

For any (possibly infinite) $\Lambda \subset \mathbb Z^2$ we denote by $\cA_\Lambda$ the unital $C^*$-subalgebra of $\cA$ generated by $\Id$ and the $a_{x,i}$ for $x\in\Lambda$ only. When $\Lambda = \{x\}$, we denote $\caA_x = \caA_{\{x\}}$. There is a natural hierarchy of inclusions: if $\tilde{\Lambda}\subseteq\Lambda$ then $\cA_{\tilde{\Lambda}}$ is  identified with the subalgebra $\cA_{\tilde{\Lambda}}\otimes \Id_{\Lambda\setminus\tilde\Lambda}$ of $\cA_\Lambda$. Let $\cF$ be the set of finite subsets of $\mathbb Z^2$. An operator $A \in \cA_\Lambda$ with $\Lambda \in \cF$ is called \emph{local}. The smallest (in the sense of set inclusion) $\Lambda \in \cF$ such that $A \in \cA_\Lambda$ is called the \emph{support} of $A$, denoted by $\mathrm{supp}(A)$. The $*$-subalgebra 
$$
\cAloc = \bigcup_{\Lambda \in \cF} \cA_\Lambda
$$
is called the algebra of local observables. By the definition of $\cA$,  $\cA\equiv\overline \cAloc$, in the $\| \cdot \|$-topology.

\paragraph{Almost local observables.}

While $\cAloc$ is extremely convenient for keeping track of where in space observables act, as we shall recall below, it is in general not invariant under the Schr\"odinger time evolution. For that reason, we work with a standard invariant algebra, originally introduced in~\cite{Scattering}, which still keeps track of the ``support center'' of observables. It can be defined as follows. Let
$$
\mathcal L = \{f : [0,\infty)\to (0,\infty)\,|\, f\text{ is bounded, non-increasing, and } \forall p>0,\, \lim_{ r\to \infty} r^p f(r)=0\}.
$$

For $x \in \mathbb Z^2$ and $r \in \mathbb N$ we denote by $B_x(r)$ the ball of center $x$ and radius $r$ within $\bbZZ^2$. An observable $A \in \cA$ is called $f$-\emph{localized} near $x \in \mathbb Z^2$ if there exists $f \in \mathcal L$, a sequence $A_n \in \cA_{B_x(n)}$ such that 
$$
\| A - A_n\| \leq f( n ) \|A\|
$$
for all $n \in \mathbb N$. An observable is \emph{almost local} iff it is $f$-localized for some function $f\in\calL$ and some $x \in \mathbb Z^2$. We denote by $\cAal$ the $*$-algebra of almost local observables. Since $\cAloc \subset \cAal$, the latter is also dense in $\cA$.

\paragraph{Parity.} The \emph{fermionic parity} is the unique $*$-automorphism $\theta:\cA\to\cA$ which satisfies $\theta(a_{x,i}) = - a_{x,i}$ for all
$x \in \mathbb Z^2$ and $i \in \{1,\ldots, n\}$. Due to $\theta$'s local structure, for any $\Lambda \subset \mathbb Z^2$ one has $\theta(\cA_\Lambda) = \cA_\Lambda$. An element $A\in \cA$ is called \emph{even} if $\theta(A)=A$, \emph{odd} if $\theta(A) =-A$, and \emph{homogeneous} if it is either odd or even. For $\Lambda \subset \mathbb Z^2$ we denote by
$
\cA^+_\Lambda = \{A \in \cA_\Lambda \, | \, \theta(A) = A\}
$
the $C^*$-subalgebra of even elements in $\cA_\Lambda$, and identify $\cA^+_{\mathbb Z^2}$ with $\cA^+$.

\paragraph{$0$-chains and autonomous dynamics.}
    Let $f\in\calL$ be given. An $f$-local 0-chain $F$~\cite{KSNoether} is a sequence $(F_x)_{x\in \mathbb Z^2}\subseteq\cAal\cap\cA^+$ such that
    \begin{itemize}
        \item For all $x \in \mathbb Z^2$, $F_x$ is self-adjoint and $f$-localized near $x$.
        \item $\displaystyle \sup_{x\in \mathbb Z^2} \| F_x \| < \infty$.
    \end{itemize}

\begin{prop}\label{prop:chain_derivation}
Let $F$ be an $f$-local 0-chain. Then the densely defined map $\delta^F:\cAal\to\cAal$ given by
$$
 A  \mapsto  \iu[F,A]:= \sum_{x\in \mathbb Z^2} \iu[F_x,A]
$$
is a *-derivation on $\cA$.
\end{prop}
\noindent In general, $\delta^F$ is an unbounded operator on $\caA$.
\begin{proof}
    Let $A$ be $f_A$-localized near some $x_A\in\bbZZ^2$ for some $f_A\in\calL$ and let $C := \sup_x \norm{F_x}$. We claim that $\delta^F(A)$ is $g$-localized near $x_A$ for some $g\in\caL$. Let $n\geq 2\norm{x-x_A}$ then
    \begin{equation*}
        [F_x,A]
        = [F_{x,n/2},A_{n}]+[F_{x},A-A_{n}]+[F_x-F_{x,n/2},A]
    \end{equation*}
with the first term being supported on $B_{x_A}(n)$ and both remaining terms are bounded above by $C\Vert A\Vert g(n)$ for some $g\in\caL$, where $g$ can be chosen to be independent of $x$. Hence $[F_x,A]$ is  $g$-localized near $x_A$. It follows that $\sum_{x\in \mathbb Z^2} [F_x,A]$ is summable and defines an element in~$\cAal$.
\end{proof}
For an $f$-local 0-chain $F$, the derivation $\delta^F$ generates a strongly continuous one-parameter group of $*$-automorphisms $\mathbb R\ni t\mapsto\alpha_t^F$ on $\cA$ which is defined by $\alpha_t^F := \exp(t \delta^F)$, i.e., as the solution to
$$\alpha^F_0(A) = A, \qquad \dfrac{\dd }{\dd t}\alpha_t^F(A) = \alpha_t^F(\delta^F(A))
$$
on $\cAal$ and extends by continuity to all of $\caA$. Furthermore, it satisfies the semi-group property $\alpha^F_{t} \circ \alpha^F_s = \alpha^F_{t+s}$. 

\paragraph{Non-autonomous setting.} The use of quasi-adiabatic flow and parallel transport below requires that we extend the formalism beyond the autonomous setting. A \emph{time-dependent $0$-chain} is a family $(F(s))_{s \in \mathbb R}$ such that:
\begin{itemize}
    \item For any $s_0 \in \mathbb R$ there exists $f\in\caL$ such that $F(s)$ is an $f$-localized 0-chain for all $|s|<s_0$. 
    \item $\displaystyle \sup_{x\in \mathbb Z^2, s \in \mathbb R} \|F_x(s)\| < \infty$.
    \item $s \mapsto F_x(s)$ is norm-continuous for all $x \in \mathbb Z^2$. 
\end{itemize}

As in the time-independent case, a time-dependent $0$-chain generates a non-autonomous time evolution $\alpha_{s\to t}^F$ for $s,t \in \mathbb R$ on $\caA$ defined by
$$\alpha^F_{s\to s}(A) = A, \qquad \dfrac{\dd }{\dd t}\alpha_{s\to t}^F(A) = \alpha_{s\to t}^F(\delta^{F(t)}(A))\qquad(A\in\cAal)\,.
$$ The semigroup property is replaced by the following cocycle property $\alpha_{t'\to t} \circ \alpha_{s\to t'}=\alpha_{s\to t}$. 

\begin{rem}
    From now on, unless stated, all the $0$-chains that we consider are time-dependent. For convenience we shall slightly abuse notation and denote $\alpha^F_{t}:=\alpha^F_{0\to t}$.
\end{rem}

The next result is a consequence of the Lieb-Robinson bound (see for example~\cite{teufel2025LRB} for a strong version that applies in our setting, and the references therein) and it is the main reason for the introduction of the algebra $\cAal$. 
\begin{prop}
   Let $F$ be a $0$-chain. Then, for all $t,s \in \mathbb R$, 
   $$
   \alpha_{s\to t}^F(\cAal) \subset \cAal.
   $$
   Moreover, $\alpha_{s\to t}^F$ is homogeneous in the sense that $\alpha_t^F \circ \theta  = \theta \circ \alpha_t^F$.
\end{prop}

\begin{defn}[LGA]
    An automorphism $\alpha$ is called a \emph{locally generated automorphism} (LGA) if there exists a $0$-chain $F$ and some $s\in[0,\infty)$ such that $\alpha = \alpha_s^F$.
\end{defn}

\paragraph{Charge and symmetry}

The relevant symmetry to consider in the quantum Hall effect is charge conservation, which was already defined in \cref{eq:U(1) on CAR} but we give here a more concrete description. The charge operator at any site $x\in\bbZ^2$ is given by
\begin{equation}\label{eq:defQx}
Q_x = \sum_{i=1}^n a^*_{x,i} a_{x,i},
\end{equation}
which accounts for internal degrees of freedom.  Since $Q_x$ is supported on the single site $x$, the collection $\{Q_x:x\in\bbZ^2\}$ defines a (time-independent) $0$-chain which is trivially $f$-localized. For $\Lambda \in\caF$ the local charge $Q^\Lambda = \sum_{x \in \Lambda} Q_x$ is supported in $\Lambda$ and it has integer spectrum. For $\phi \in \mathbb R$ we denote by $\rho_\phi = \alpha_\phi^Q$ and $\rho_\phi^\Lambda = \alpha_\phi^{Q^\Lambda}$ the LGA respectively associated to $Q$ and $Q^\Lambda$.
\begin{prop}\label{prop:U(1) trans}
For any $\Lambda \in \cF,\phi\in\mathbb{R}$ and $A\in\caA$, one has $
\rho_\phi^\Lambda(A) = \ee^{\Idi \phi Q^\Lambda} A \ee^{-\Idi \phi Q^\Lambda}$. In particular, $\rho^\Lambda_{\phi+2\pi}=\rho^\Lambda_{\phi}$. Moreover,
\begin{equation*}
    \rho_\phi(A) = \lim_{\Lambda \to \mathbb Z^2} \ee^{\Idi \phi Q^\Lambda} A \ee^{-\Idi \phi Q^\Lambda}
\end{equation*}
for all $\phi\in\mathbb{R}$ and $\rho_{\phi+2\pi}=\rho_{\phi}$.
\end{prop}
\begin{proof}
The periodicity follows from the integrality of the spectrum of $Q^\Lambda$ and the group property of $\ee^{\Idi \phi Q^\Lambda}$. Existence of the limit is immediate for any $A\in\cAloc$ since the sequence is eventually constant. The general case follows by density.
\end{proof}
 Notice that the same limiting procedure yields an automorphism $\rho^\Lambda_{\phi}$ for any other infinite subset of $\bbZ^2$. Below, we shall in particular consider $\Lambda$ to be a half-plane.

The LGA $\rho_\phi$ is the $U(1)$- (equivalently $\mathbb{S}^1$-) transformation associated with charge conservation. A $0$-chain $F$ is called \emph{$Q$-preserving} (or \emph{$U(1)$-symmetric}) if $\rho_\phi(F_x(s)) = F_x(s)$ for all $\phi \in \mathbb R$, $x \in \mathbb Z^2$ and $s \in \mathbb R$. If $F$ is a  $Q$-preserving $0$-chain, then $\alpha_{s\to t}^F$ is such that $\rho_\phi \circ \alpha_{s\to t}^F  = \alpha_{s\to t}^F\circ \rho_\phi$ for all $s,t,\phi \in \mathbb R$. Consequently, an LGA is called \emph{$Q$-preserving} if there exists a $Q$-preserving $0$-chain which generates it. 

\subsection{Symmetric SRE  states}

A state $\omega_0 \in \cS(\cA)$ is called a \emph{product state} iff 
$$
\forall X,Y \subset \mathbb Z^2, \quad X\cap Y = \emptyset, \quad A\in\caA_X,B\in\caA_Y, \qquad \omega_0(A B) = \omega_0(A) \omega_0(B).
$$

\begin{lem}
    Let $\omega_0$ be a pure state. Then $\omega_0$ is a product state if and only if its restriction to $\cA_x$ is pure for any $x \in \mathbb Z^2$.
\end{lem}

\noindent We refer the reader to~\cite{ArakiMoriya2003} for the proof.

\begin{rem}
    Pure product states are the simplest states on $\cA$ in that they have no entanglement between any two disjoint sets of sites, and they are convenient representatives of the trivial phase. By analogy, in the non-interacting, single particle picture, they correspond to ground states of tight-binding Hamiltonians where all lattice sites are disconnected --- the atomic limit. Mathematically these correspond to Fermi projections $P$ which are diagonal in the position basis $P_{xy}=P_{xx}\delta_{xy}$. Indeed, in that case the quasi-free state $\omega_P$ is product.
\end{rem}

\begin{defn}[SRE]
A state $\omega \in \cS(\cA)$ is \emph{short range entangled} (SRE) iff there exists a pure product state $\omega_0$ and an LGA $\alpha$ such that $\omega= \omega_0 \circ \alpha$.
\end{defn}

 Note that an SRE state is necessarily pure, {since automorphisms preserve purity.} 
 
 \begin{defn}
 	A state $\omega \in \cS(\cA)$ is called \emph{$Q$-symmetric} iff $\omega \circ \rho_\phi = \omega$ for all $\phi \in \mathbb R$. We denote $\cSs^Q(\cA)$ the set of $Q$-symmetric states.
 \end{defn}

We remark that, with $Q$ defined as in~\cref{eq:defQx}, a $Q$-symmetric state is necessarily homogeneous, namely $\omega\circ\theta = \omega$ since $\theta = \ep{\iu\pi Q}$.

\begin{defn}
	Let $\omega$ be a pure state and $H$ be a $0$-chain. We say that $\omega$ is a ground state of $H$ iff $\omega(A^*\delta^H(A))\geq0$ for all $A \in \cAloc$. Moreover we say that $\omega$ is a {locally} unique gapped ground state of $H$ with gap $\Delta>0$ if
	$$
	\omega(A^*\delta^H(A)) \geq \Delta \omega(A^*A)
	$$
	for all $A \in \cAloc$ such that $\omega(A)=0$.
\end{defn}

If $\omega$ is a pure state, we say that a 0-chain $H$ is a \emph{parent Hamiltonian} for $\omega$ if $\omega$ is a locally unique gapped ground state of~$H$. An essential, although simple, result of this section is the following lemma, whose proof may be found in~\cite[Lemma 3.3]{BachmannBolsRahnama24}.
\begin{lem}\label{lem:Q parent}
	Any $Q$-symmetric SRE state has a $Q$-preserving parent Hamiltonian.
\end{lem}

\begin{rem}
    In the present work, a parent Hamiltonian does not need to have a physical interpretation, but it exists as a mathematical tool. This emphasizes the fact that we are classifying states, not Hamiltonians.
\end{rem}

Last but not least, we introduce the relevant notion of `deformation' of states, which is used to define the topological phases. We remind the reader of the overall symmetry assumption introduced above.
\begin{defn}
	Two pure and $Q$-symmetric states $\omega_1, \omega_2$ are called \emph{$Q$-equivalent} if there exists a $Q$-preserving LGA $\alpha$ on $\cA$ such that $\omega_2 = \omega_1 \circ \alpha$. We shall denote this $\omega_1\sim\omega_2$.
\end{defn}

\begin{rem}
	A symmetric SRE state is not necessarily $Q$-equivalent to a product state because we do not assume the LGA in the definition of SRE to be $Q$-preserving. The SRE condition encodes the locality of the state, regardless of the charge symmetry, whereas equivalence describes $Q$-preserving deformations of states. Symmetric SRE states which are not $Q$-equivalent to a product state are called symmetry-protected topologically (SPT) ordered~\cite{wen2017colloquium}.
\end{rem}

%%%%%%%%%%%%%%%%%%%%%%%%%%%%%%%%%%%%%%
\subsection{Magnetic flux insertion at the origin \label{sec:flux_insertion}}

To probe the Hall conductivity associated with a state $\omega$, we follow the Laughlin picture and focus on charge deficiency arising from a magnetic flux insertion in our system. We will show, see also~\cite{bachmann2020many,kapustin2020hall}, that the Hall conductivity of $\omega$ is equal to $\calN_\rho(\omega,\omega^D)$ where $\omega^D$ is obtained from $\omega$ by inserting one unit of magnetic flux.

Before we proceed with the technical construction, we make a few comments to motivate what is ahead. Following the work of \cite{AvronSeilerSimon_Charge}, in the many-body setting we would like to create $\omega^D$ out of $\omega$ by analogy to how $L^\ast P L$ is obtained out of the Fermi projection $P$, where the operator $L$ is the Laughlin unitary given in \cref{eq:Laughlin unitary}. For quasi-free states, this lifting is done by considering the outer automorphism $\ell:\calA\to\calA$ via \eq{
a(f) \mapsto a(Lf)\qquad (f\in \caH)
} and extend linearly, and recalling that $\omega_P\circ\ell=\omega_{L^\ast P L}$. Hence, also for non-quasi-free-states, the natural choice is then $\omega^D := \omega\circ \ell$ with which one may hope to show that, under suitable assumptions on the initial state $\omega$, the pair $\omega,\omega^D$ is $\rho$-locally-comparable and thus define the charge deficiency index of $\omega$ as
$\calN_\rho(\omega,\omega\circ\ell)$.

As proposed, this is impossible. Indeed, $P-L^\ast P L \in \caJ_3(\caH)$ but the analysis in \cref{subsec:IPP} indicates that if $P-Q$ is not Hilbert-Schmidt then $(\omega_P,\omega_Q)$ cannot be $\rho$-locally-comparable. In fact we made the even stronger assumption $P-Q\in \caJ_1(\caH)$.

This is remedied by considering another gauge (the `half-line gauge') for the same flux insertion procedure. For non-interacting fermions, this was studied in detail in~\cite{DeNittisSchulzBaldes2016spectral}. In the interacting setting, this was carried out in~\cite{BachmannBolsRahnama24}, following the closely related~\cite{kapustin2020hall}, where it is shown that the corresponding automorphism applied to an initial SRE state yields a unitarily equivalent state. This extends to so-called invertible states, see \cref{def:invertible} below. Interestingly, in general the unitary depends on the initial state (unlike the universal $L$), and we shall see that continuity properties of $\caN(\omega,\omega^D)$ must be dealt with carefully. In particular, there is no inner automorphism that directly implements a half-line gauge. 

Fundamentally, the implementability of the flux insertion procedure for invertible states is due to the triviality of their superselection sectors~\cite{Triviality}. Beyond the invertible setting, one expects flux insertion to produce states that are truly inequivalent to the initial state in the sense that they carry anyonic quasi-particles~\cite{Categories}.

For the sake of clarity, we consider the upper half-plane $\plane^\up := \mathbb Z \times \mathbb N$, but in principle one could consider any half-plane, see e.g. \cite{BachmannBolsRahnama24}.

 \paragraph{$U(1)$-transformations on half-planes.} We consider the $U(1)$-transformation $\rho^\up := \rho^{\plane^\up}$ associated with the upper half-plane. Let $H$ be a $Q$-preserving $0$-chain. For $\phi \in \mathbb R$ we consider the 0-chain $H^\up(\phi)$ defined by
 $$ \forall x \in \mathbb Z^2,\qquad H^\up(\phi)_x = \rho^\up_{-\phi} (H_x)\,.$$
 Notice that $H^\up(\phi)$ is $Q$-preserving for all $\phi$.

 \begin{lem}
     Let $\omega$ be an SRE state and let $H$ be a $Q$-preserving parent Hamiltonian. Then $H^\up(\phi)$ is a $Q$-preserving parent Hamiltonian of $\omega^\up_\phi =  \omega\circ\rho^\up_\phi $ for all $\phi \in \mathbb R$.
 \end{lem}

Note that the state is `rotated forward' while the Hamiltonian is `rotated backwards'; see also the proof in~\cite[Section 3.1.3]{BachmannBolsRahnama24}.

\paragraph{Quasi-adiabatic continuation} The fact that $\omega$ is invariant under $U(1)$-transformation indicates that $\omega^\up_\phi$ differs from $\omega$ only along the boundary of the upper half-plane. This intuition can be made precise by using the fact that $\omega$ has a gapped Hamiltonian: This is the role of the spectral flow~\cite{hastingswen, AutomorphicEq}.

 \begin{prop}\label{prop:defK}
     There exists a $0$-chain $K=(K(\phi))_{\phi\in \mathbb R}$ such that
     \begin{equation}\label{eq:halfspaceU1_implementation}
     \omega^\up_\phi = \omega \circ \alpha_\phi^K.
     \end{equation}
     The generator satisfies $K_x(\phi) = \theta(K_x(\phi)) = \rho_{\phi'}(K_x(\phi))$ for all $x\in \mathbb Z^2$ and $\phi,\phi' \in \mathbb R$. Moreover, there exists a function $g \in \mathcal L$ such that
     \begin{equation}\label{eq:K_local_at_boundary}
     \forall x=(x_1,x_2) \in \mathbb Z^2,\quad \forall \phi \in \mathbb R, \qquad \| K_x(\phi) \| \leq  g(|x_2|).
     \end{equation}
\end{prop}
     The proof is in~\cite[Section 3.1.3]{BachmannBolsRahnama24}, but we recall here that 
\begin{equation}\label{eq:Def of K}
    K_x(\phi) := \int_{\mathbb R} \dd t W(t) \alpha_t^{H^\up(\phi)} \left(\dfrac{\dd }{\dd \phi} [H^\up(\phi)_x]\right)
\end{equation}
     where $W : \mathbb R \to \mathbb R$ is an odd bounded function decaying faster than any power at $t\to \pm \infty$ and with Fourier transform $\hat W(E)$ equal to $-\Idi/E$ for $|E|>1$. Notice that for a given $\phi\in \mathbb R$, $\alpha_t^{H^\up(\phi)}$ is the LGA generated by the $t$-independent generator $H^\up(\phi)$.

 This proposition is essential: equation \eqref{eq:halfspaceU1_implementation} states that $U(1)$-transformations on half-spaces can be implemented on the state $\omega$ by another LGA, which is non-autonomous but acts non-trivially only along the boundary of the upper half-space, see~\eqref{eq:K_local_at_boundary}.

 \begin{defn}[Magnetic flux insertion]\label{def:flux_insertion_autom}
     Let $K$ be as in~\cref{eq:Def of K}. For $\phi \in \mathbb R$ and $x=(x_1,x_2) \in \mathbb Z^2$ we define
     $$
     D_x(\phi) =K_x(\phi) \mathds{1}_{\{x_1<0\}} := \begin{cases}
         K_x(\phi), & x_1 <0, \\
         0, & \text{otherwise.}
         \end{cases}
     $$ 
     The associated LGA, which we denote $\D_\phi = \alpha^D_\phi$, is called the flux insertion automorphism. 
\end{defn}

\begin{rem} By construction $D(\phi)$ is $g$-localized near the half-line $\{x_1 <0, x_2=0\}$. It corresponds to a defect line (or branch cut) of vector potential associated with a $\phi$-flux  insertion at the origin. The dynamics $\D_\phi$ is a caricature of the adiabatic switching of a magnetic flux from $0$ to $\phi$.
\end{rem}

 \begin{defn}[Defect state]
     Let $\omega$ be an SRE state. The state $\omega^D_\phi := \omega \circ \D_\phi$ is called the \emph{defect state} at flux $\phi$. At $\phi = 2\pi$, we denote $\omega^D = \omega^D_{2\pi}$. 
\end{defn}

We recall that the construction from the symmetric SRE state $\omega$ to the defect state $\omega^D$ is summarized in Figure~\ref{fig:roadmap}.

\begin{rem}
    Since $K$ is even and $Q$-preserving, so is $D$, and so $
     \D_\phi \circ \rho_{\phi'} = \rho_{\phi'} \circ \D_\phi$ as well as $\D_\phi \circ \theta = \theta \circ \D_\phi$. This implies immediately that $\omega^D_\phi$ is $Q$-symmetric. 
\end{rem}

The two states $\omega,\omega^D_\phi$ differ from each other along the half-line $\{x_1 <0, x_2=0\}$. At $\phi = 2\pi$, the intrinsic locality of the state, namely the SRE assumption, implies that, in fact, $\omega,\omega^D$ differ from each other only around the flux insertion point. Indeed, far along the half-line, $\omega^D \simeq \omega\circ\rho_{2\pi} = \omega$ because $\rho$ is a $U(1)$-transformation. This implies that $(\omega,\omega^D)$ are locally comparable. In fact, they are even $\rho$-locally comparable because the unitary can be shown to be almost localized near the origin, see \cref{def:comp} and \cref{def:alpha comp}. This is the content of the following proposition.

\begin{prop}\label{prop:existence_of_U}
    Let $\omega\in \cSs^Q(\cA)$ be SRE and let $\omega^D$ be the associated defect state with flux~$2\pi$. Then there exists a homogeneous unitary $u\in \cAal$ such that $\omega^D = \omega \circ \mathrm{Ad}_u$. Moreover, $u\in\caD(\delta^Q)$.
\end{prop}

The proof of the existence of $u$ with all claimed properties except for the last one can be found in~\cite[Proposition 3.8]{BachmannBolsRahnama24}, with the replacements of $\omega_{-\pi}$, respectively $\omega_{\pi}$, there by $\omega$, respectively $\omega^D$, here. See also~\cite[Appendix~A]{KapustinArtymowicz}. Finally, $u\in\cAal$ implies that it lies in the domain of $\delta^Q$, see \cref{prop:chain_derivation}.
\begin{rem}\label{Rem:SRE U}
    (i) $u$ is the many-body analog of the Laughlin unitary discussed above. It exists only for $\phi=2\pi$ (and could also be constructed for other integer fluxes).\\
    (ii) If $\omega$ is a symmetric product state, then one can take $u=\Id$. \\
    (iii) The fact that $u$ implements the $Q$-preserving transformation $\gamma_{2\pi}$ on the symmetric state $\omega$ does not imply that $\delta^Q(u) = 0$. In \cref{example:shift}, the bilateral shift reads in second quantization $\alpha^R(a_x) = a_{x+1}$ for all $x\in\bbZ$, while $\rho_\phi(a_x) = \ep{-\iu\phi}a_x$. Clearly, we have that $\alpha^R\circ\rho_\phi = \rho_\phi\circ\alpha^R$. However, with $u = a_0 + a_0^*$ given there, we see that
    \begin{equation*}
        -\iu \delta^Q(u) = -\iu \left.\frac{d}{d\phi}\rho_\phi(u)\right\vert_{\phi=0} = -a_0 + a_0^*.
    \end{equation*}
\end{rem}

\begin{prop}\label{prop:SRE_index}
   Let $\omega\in \cSs^Q(\cA)$ be SRE and let $u\in \cAal$ be the unitary given by \cref{prop:existence_of_U}. Let 
    \begin{equation}\label{eq:def_small_i}
        i(\omega,Q) := \iu\omega(u\str \delta^Q(u))\,.
    \end{equation}
    If $\beta$ is a $Q$-preserving LGA then  $i(\omega,Q) =i(\omega\circ\beta,Q)$. Moreover, $i(\omega,Q)$ is independent of the choice of symmetric parent Hamiltonian $H$ for $\omega$.
\end{prop}

\begin{proof}
The definition~\cref{eq:def_small_i} is nothing else than the abstract index~\cref{eq:definition of general index} applied in the present specific context, namely
\begin{equation*}
    i(\omega,Q) = \caN_{\rho}(\omega,\omega^D).
\end{equation*}
An LGA deformation of the initial state $\omega(s) = \omega\circ\beta_s$ yields a deformation of the flux insertion automorphism $\D_s$ and hence a family of states $\omega^D(s) = \omega(s) \circ \D_s$. If $\omega$ is a symmetric SRE state, so is $\omega(s)$ and so there is a family $u_s$ such that $\omega^D(s) = \omega(s)\circ\mathrm{Ad}_{u_s}$. In particular, $(\omega^D(s), \omega(s))$ are $\rho$-locally comparable for all $s$. Note that, in general, $\omega$ and $\omega(s)$ (as well as their defect counterparts) are not locally-comparable, even for infinitesimal $s$. However, the state 
\begin{equation*}
    \widetilde\omega(s) = \omega\circ\beta_s\circ\mathrm{Ad}_{u_s}\circ\beta_s^{-1}
\end{equation*}
can be expressed as $\widetilde\omega(s) = \omega\circ\mathrm{Ad}_{\beta_s(u_s)}$, so that $(\widetilde\omega(s), \omega)$ are $\rho$-locally comparable for all $s$ and $i(\omega,Q) = \caN_{\rho}(\omega,\tilde\omega(0))$. The family of normalized vectors $\Omega_{\widetilde\omega(s)} = \pi_{\omega}(\beta_s(u_s))\Omega_{\omega}$ are vector representatives of the state $\widetilde\omega(s)$ in the GNS Hilbert space $\caH_\omega$ of $\omega$. If we again let $Q_\omega$ be the charge operator in the GNS representation, then $\Omega_{\widetilde\omega(s)}$ are all eigenvectors of $Q_\omega$. We claim that $s\mapsto\Omega_{\widetilde\omega(s)}$ is a norm-continuous family of vectors in $\caH_\omega$. Since orthogonal unit vectors $\Psi,\Phi$ have $\Vert\Psi - \Phi\Vert^2 =2$, we conclude that $\Omega_{\widetilde\omega(s)}$ are in a constant eigenspace. This in turn means that $\caN_{\rho}(\omega,\tilde\omega(s))$ is constant. We conclude  that
\begin{equation*}
    i(\omega,Q) = \caN_{\rho}(\omega,\tilde\omega(0))
    = \caN_{\rho}(\omega,\tilde\omega(s))
    = \caN_{\rho}(\omega(s),\omega^D(s))
\end{equation*}
where the last equality is by~\cref{thm:PropertiesGeneralIndex}(iv).

It remains to prove the continuity of $\Omega_{\widetilde\omega(s)}$. We first claim that $s\mapsto\D_s$ is strongly continuous. If $\omega = \omega_0\circ\alpha$, the family $\omega(s)$ is also SRE, with entangler $\alpha\circ \beta_s$. Recalling the summary~\cref{sec:summary}, we see that the parent Hamiltonian $H_s$ is simply given by $H_{s,x} = \beta_s^{-1} (H_x)$ and the interaction terms are continuous. Also, the family $H_s$ has constant spectral gap. It is now a consequence of the Lieb-Robinson bound, see e.g.~\cite[Theorem 6.2.11]{BratteliRobinson2}, that the dynamics $\alpha_t^{H^\uparrow_s(\phi)}$ is strongly continuous with respect to $H^\uparrow$. In particular, the interaction terms defined by
$$
     K_{s,x}(\phi) := \int_{\mathbb R} \dd t W(t) \alpha_t^{H_s^\up(\phi)} \left(\dfrac{\dd }{\dd \phi} [H^\up(\phi)_{s,x}]\right)
     $$
are continuous as a function of $s$. So is restriction $D_s$ of the $0$-chain $K_s$ to the half-line, and the claimed strong continuity of $s\mapsto\D_{s,\phi}$ follows again by the Lieb-Robinson bound. With this, the formula
\begin{equation*}
    \widetilde\omega(s)
    = \omega\circ\beta_s\circ\D_{s,2\pi}\circ\beta_s^{-1}
\end{equation*}
implies immediately the weak-* continuity of $s\mapsto \widetilde\omega_2(s)$. Hence $s\mapsto\Omega_{\widetilde\omega_2(s)}$ is weakly continuous (in $\caH_\omega$) and since this is a family of constant norm, it is in fact norm continuous, concluding that part of the proof.

   We turn to the invariance under the choice of parent Hamiltonian. For all $\phi\in[0,2\pi]$, let $\omega_\phi^{D,j} = \omega\circ\D^{j}_\phi, j=1,2$ be the defect states generated by the Hamiltonians $H_j$. We claim that, for all $A\in\caA$, 
   \begin{equation}\label{eq:local comparability of tilde D}
       \vert \omega_0(A) - \tilde \omega^{D}_\phi(A)\vert\leq f(r)\Vert A \Vert
   \end{equation}
   where $\tilde \omega^D_\phi = \omega_0\circ\alpha\circ\D^1_\phi\circ(\D_\phi^2)^{-1}\circ\alpha^{-1}$.
   Indeed, by the locality of $\alpha$, we can replace $\D_{\phi}^{1} \circ \big( \D_{\phi}^{2} \big)^{-1}$ by $\alpha_{\phi}^{K_1} \circ \big( \alpha_{\phi}^{K_2} \big)^{-1}$ for all $\phi\in[0,2\pi]$ and all $A$ supported in~$(\Lambda\cap B_0(r)^c)$ or in $(\Lambda^c\cap B_0(r)^c)$, where $\Lambda$ is the left half-plane. This is because, far away from the origin, the effect of the truncation $K\to D$ is irrelevant in the left half-plane and similarly on the right half-plane simply because both $\gamma_\phi^D$ and $\alpha_\phi^K$ act as identity, see~\cite[Lemma~3.7]{BachmannBolsRahnama24}. This, and the identity
    \begin{equation*}
        \omega\circ\alpha_{\phi}^{K_1} \circ \big( \alpha_{\phi}^{K_2} \big)^{-1}
        = \omega_\phi^{\uparrow} \circ \big( \alpha_{\phi}^{K_2} \big)^{-1}
        = \omega,
    \end{equation*}
    implies~\cref{eq:local comparability of tilde D} for all $\phi\in[0,2\pi]$ and all $A$ supported in~$(\Lambda\cap B_0(r)^c)$ or in $(\Lambda^c\cap B_0(r)^c)$. That this can be lifted to all $A\in\caA$ follows from~\cite[Appendix~A]{BachmannBolsRahnama24}, and in turn there is a unitary $w_\phi$ such that
    \begin{equation*}
        \tilde \omega^D_\phi = \omega_0\circ\mathrm{Ad}_{w_{\phi}}.
    \end{equation*} 
    The family of $Q$-symmetric states $\tilde \omega^D_\phi$ is weakly-* continuous and so 
    \begin{equation*}
        \caN_{\rho}(\omega_0, \tilde \omega^D_\phi) = \caN_{\rho}(\omega_0,\omega_0) = 0
    \end{equation*}
    by the continuity argument in the first part of the proof, with $\omega\to\omega_0$ and $\tilde\omega(s)\to\tilde \omega^D_\phi$, and \cref{thm:PropertiesGeneralIndex}(v). Evaluating this at $\phi = 2\pi$, we conclude by~\cref{thm:PropertiesGeneralIndex}(iv) that
    \begin{equation*}
        0 = \caN_{\rho}(\omega_0, \tilde \omega^D_{2\pi}) 
        = \caN_{\rho}(\omega^{D,2}, \omega^{D,1}),
    \end{equation*} 
    and by additivity
    \begin{equation*}
        \caN_{\rho}(\omega,\omega^{D,1})
        =\caN_{\rho}(\omega,\omega^{D,2}) + \caN_{\rho}(\omega^{D,2},\omega^{D,1})
        =\caN_{\rho}(\omega,\omega^{D,2})
    \end{equation*}
    which is what we had set out to prove.
\end{proof}

\begin{rem}
    If $\omega$ is a symmetric product state, then the $Q$-preserving parent Hamiltonian of \cref{lem:Q parent} can be chosen to be purely on-site. Since the $U(1)$-transformation is also on-site, the Hamiltonian is invariant under $\rho^{\uparrow}$. Hence the quasi-adiabatic generator $K$ of~(\ref{eq:Def of K}) is uniformly equal to $0$. It follows that $\omega^D = \omega$, so that $u$ can be taken to be the identity and we conclude that $i(\omega,Q) = 0$ for an initial product state, as it should.
\end{rem}

As we shall see in \cref{sec:free} below, the problem with  \cref{prop:SRE_index} is that $i(\omega,Q)=0$ for any quasi-free SRE state $\omega$ describing the (non-interacting) quantum Hall effect. In the next section, we extend the definition of the index to the larger class of invertible states which covers all topological ground states of non-interacting models. In principle, other physical models might provide examples of many-body SRE states with a non-trivial index, without stacking with an auxiliary system, but we are not aware of them.

\subsection{Many-body index \label{sec:MBI}}

So far, we have shown how to associate an index to an SRE state in the presence of a $U(1)$-symmetry. The SRE assumption is used in two distinct steps: first to ensure that the state has a gapped parent Hamiltonian, and second to show that the associated defect state is $\rho$-locally comparable with the initial state. In both cases, the SRE assumption is a sufficient condition but in no way necessary. In this section, we extend the index to invertible states.

\paragraph{Stacking.} In the following we shall make use of \emph{stacking} which is a procedure to increase the number of internal degrees of freedom. For this, let $\cA^1$ and $\cA^2$ respectively be the CAR algebras over $\ell^2(\mathbb Z^2)\otimes \mathbb C^n$ and $\ell^2(\mathbb Z^2)\otimes \mathbb C^m$ for $n,m \in \mathbb N$. We denote by $\cA^1 \hat\otimes \cA^2$ the CAR algebra over $\ell^2(\mathbb Z^2)\otimes \mathbb C^{n+m}$. 

\begin{rem}
	{A word about $\hat{\otimes}$: The standard construction of the tensor product $\cA^1\otimes\cA^2$ is also a CAR algebra, but to see that one has to reshuffle its generators. Indeed, consider the simple example of space consisting of a single point and $n=m=1$. Then $a,\tilde{a}$, the two generators of $\cA^1,\cA^2\cong\mathrm{Mat}_2(\bbC)$ respectively, embed naturally into $\cA^1\otimes\cA^2$ as $a\otimes\Id,\Id\otimes\tilde{a}$ respectively, which commute rather than anti-commute. A choice that realizes the CAR is $a_1 = a \otimes \Id$ and $a_2 = (-1)^{a\str a}\otimes \tilde a$. We use the symbol $\hat{\otimes}$ to denote this reshuffling of the generators (though the underlying C*-algebraic structure is, of course, identical).}
\end{rem}

The stacking operation on algebras naturally extends to states: for $\omega_1 \in \cS(\cA^1)$ and $\omega_2 \in \cS(\cA^2)$ we define $\omega_1 \hat \otimes \omega_2 \in \cS(\cA^1 \hat \otimes \cA^2)$ by \eq{(\omega_1 \hat \otimes \omega_2)(A_1 \hat \otimes A_2) :=\omega_1(A_1)\omega_2(A_2)\qquad(A_1\in\cA^1,A_2\in\cA^2)} and extended beyond simple tensors by linearity.

\begin{defn}[invertible and stably SRE]\label{def:invertible}
	A state $\omega \in \cS(\cA)$ is \emph{invertible} iff there is an auxiliary system $\cA'$ (namely a CAR algebra over $\ell^2(\mathbb Z^2) \otimes \mathbb C^m$ for some $m\in\bbN$) and a state $ \omega' \in \cS(\cA')$ such that $\omega \hat \otimes \omega'$ is SRE on $\cA \hat\otimes \cA'$. If $\omega'$ can be taken to be a product state, then $\omega$ is called \emph{stably SRE}.
\end{defn}

Notice that a stably SRE state is necessarily invertible. The converse is true in one dimension~\cite{kapustin2021classification}. That these two classes are not equal in two dimensions is a corollary of the results presented below, see \cref{sec:free}.

\begin{rem}
	(i) The relevance of invertible states lies in the observation that they correspond to integer quantum Hall systems. Any quasi-free state is invertible and so all non-interacting models have invertible ground states. However, it is believed that a non-vanishing Hall conductance is an obstruction to being stably SRE. Moreover, as \cref{thm:main} below shows, they do not cover fractional quantum Hall states. In fact, it can be shown that invertibility guarantees the absence of anyons, see~\cite{Triviality}.
	
	\noindent (ii) The stacking with product states, which appears in the notion of stably SRE, is analogous to the adding of some trivial bands in the single-particle setting. It also parallels a common procedure in K-theory. 
\end{rem}

Let $\omega$ be a state on $\caA$ which is invertible, with inverse $\omega'$ on auxiliary algebra $\cA'$. We denote by $\hat \omega= \omega \hat \otimes \omega'$ and $\hat \cA = \cA \hat \otimes \cA'$. The algebra $\cA'$ is equipped with a charge $Q$ as well, so we are left with three distinct charges on $\hat \cA$:
\begin{equation}\label{eq:defQstack}
Q \to Q \hat \otimes \Id, \qquad Q' = \Id \hat \otimes Q, \qquad \hat Q = Q+Q'.
\end{equation}
Notice that these three charges commute. 

\begin{defn}
	We say that $\omega$ is symmetric invertible if it is invertible with inverse $\omega'$ and if $\hat \omega = \omega \hat \otimes \omega'$ is both $Q$- and $Q'$-symmetric.
\end{defn}

In particular if $\omega$ is symmetric invertible then $\hat \omega$ is also $\hat Q$-symmetric. The parent Hamiltonian $\hat H$ provided by \cref{lem:Q parent} is $\hat Q$-symmetric. Defining then
    \begin{equation*}
        \tilde H_x = \frac{1}{4\pi^2}\int_0^{2\pi}\int_0^{2\pi} \rho^{Q'}_{\phi_2}\circ\rho_{\phi_1}^Q(\hat H_x) \dd\phi_1\dd\phi_2
    \end{equation*}
    for all $x\in\bbZ^2$ yields a parent Hamiltonian that is both $Q$- and $Q'$- preserving --- a fortiori still $\hat Q$-preserving. Note that the commutativity of $Q,Q',\hat Q$ is essential here.

    With this, we can carry out the construction described in \cref{sec:flux_insertion} and define the flux insertion automorphism $\hat \D_\phi$ and a defect state $\hat \omega_\phi^D:= \hat\omega\circ \hat \D_\phi$, with respect to the total charge $\hat Q$. By inspecting the formula~\cref{eq:Def of K}, we conclude that $\hat \D_\phi$ is invariant under both $Q,Q'$, and so the defect state $\hat \omega_\phi^D$ is $Q,Q',\hat Q$-symmetric. Hence we have the following
\begin{prop}\label{prop:stackeddefectstate_symmetry}
	A symmetric invertible state $\hat \omega$ has a symmetric parent Hamiltonian which is both $Q$- and $Q'$-symmetric. Moreover the defect state $\hat \omega^D$, generated by $\hat Q$, remains $Q$- and $Q'$-symmetric.
\end{prop}
    
    Continuing as in \cref{sec:flux_insertion}, \cref{prop:existence_of_U} yields $\hat u \in \hat \cA^\mathrm{al}$ such that $\hat \omega^D = \hat \omega \circ \mathrm{Ad}_{\hat u}$. Since $\hat u\in\hat \caA^{\rm al}$, we further have that $\hat u\in\caD(\delta^{ Q})$, namely $(\hat\omega^D,\hat\omega)$ are $\rho^Q$-locally comparable. 

\begin{defn}\label{def: THE Index}
    Let $\omega$ be an invertible state as above. The \emph{index of $\omega$} is defined by 
    \begin{equation}\label{eq:invertibleindex}
        \mathcal I(\omega) :=
        \caN_{\rho^Q\hat\otimes\mathrm{id}}(\hat\omega,\hat\omega^D) = \ii \hat \omega\left(\hat u^* \delta^{Q\hat\otimes \Id}(\hat u)\right)
    \end{equation}
\end{defn}
We display the second expression to emphasize that the charge used in this definition is on the first subsystem $\caA$, on which $\omega$ is defined, although the flux insertion was carried out using the full charge $\hat Q$.

With these preparations, we can now state our main result.

\begin{thm}\label{thm:main}
    Let $\omega$ be a symmetric invertible state. The index $\mathcal I(\omega)$ of~\cref{def: THE Index} has the following properties:
    \begin{enumerate}
        \item \emph{Integrality}:  $\mathcal I(\omega) \in \mathbb Z$.
        \item \emph{Well-definedness}: $\mathcal I(\omega)$ is independent of the choice of symmetric parent Hamiltonian for $\hat \omega$.
        \item \emph{Continuity}: If $\tilde \omega$ is $Q$-equivalent to $\omega$, then $\tilde \omega$ is symmetric invertible and $ \mathcal I(\omega) = \mathcal I(\tilde \omega)$.
        \item \emph{Additivity}: If $\omega_1$ and $\omega_2$ are symmetric invertible states on $\cA^1$ and $\cA^2$ then $\omega_1 \hat\otimes \omega_2$ is a symmetric invertible state on $\cA^1 \hat\otimes \cA^2$ and $\mathcal I(\omega_1 \hat\otimes \omega_2) = \mathcal I(\omega_1)+\mathcal I(\omega_2)$.
        \item \emph{Surjectivity}: For all $n \in \mathbb Z$, there exists an invertible state $\omega$ such that $\mathcal I(\omega) = n$.
    \end{enumerate}
\end{thm}

\begin{rem}
    (i) If the stacked state $\omega\hat \otimes \omega'$ is not only SRE, but it is $\hat Q$-equivalent to a pure product state, then 
    \begin{equation*}
        \mathcal I(\omega')=-\mathcal I(\omega)
    \end{equation*}
    by \textit{(iii)} and \textit{(iv)} of \cref{thm:main}. \\
    (ii) At this point, the injectivity of the index remains an open question, namely whether two invertible states with the same index are necessarily $Q$-equivalent.
\end{rem}

\begin{proof} 
Since
$$
\mathcal I(\omega)= 
        \caN_{\rho^{Q\hat \otimes \Id}}(\hat\omega,\hat\omega^D)  =  i(\hat\omega,Q\hat \otimes \Id),
$$
 the theorem is a direct consequence of \cref{thm:PropertiesGeneralIndex} for the abstract index and \cref{prop:SRE_index} on $\hat \cA$. We check that the assumptions hold in the context of invertible states.
 
        \textit{(i)} The states $\hat \omega$ and $\hat \omega^D = \hat \omega \circ \mathrm{Ad}_{\hat u}$ are $\hat \rho$-locally-comparable since $\hat u \in \hat \cAal$ and $Q = Q \hat \otimes \Id$ is a $0$-chain of $\hat \cA$. The state $\hat \omega = \omega \hat\otimes \omega'$ is $Q$-symmetric since $\omega$ is symmetric invertible. Moreover, $\hat \omega^D$ is $Q$-symmetric as well by \cref{prop:stackeddefectstate_symmetry}. Thus $\mathcal I(\omega) = \mathcal N_{\rho^Q}(\hat \omega, \hat \omega^D) \in \mathbb Z$ by \cref{thm:PropertiesGeneralIndex}(ii).

        \textit{(ii)} This follows as in the proof of~\cref{prop:SRE_index}.
        
        \textit{(iii)} Let $\tilde \omega$ be $Q$-equivalent to $\omega$. Let $\beta$ be a $Q$-preserving LGA on $\cA$ such that $\tilde \omega = \omega \circ \beta$. Since $\omega$ is invertible, let $ \omega'$ be its inverse. One has $\omega \hat \otimes  \omega' = \hat \omega_0 \circ \hat \alpha$ with an LGA $\hat \alpha$ and a pure product state $\hat \omega_0$ on $\hat \cA$. Thus we have
        $$
        \tilde \omega \hat \otimes  \omega'  = \hat \omega_0 \circ \hat \alpha \circ (\beta \hat\otimes \Id)\,.
        $$
        Thus $\tilde \omega$ is invertible. Moreover it is $Q$-symmetric and since $\hat \omega$ is symmetric invertible, $ \omega'$ is $Q'$-symmetric. Thus $\tilde \omega$ is symmetric invertible, so that $\mathcal I(\tilde \omega)$ is well defined. The invariance of the index $\mathcal I(\tilde \omega) = \mathcal I(\omega)$ follows from \cref{thm:PropertiesGeneralIndex}(iv) on $\hat \cA$ and goes along the same lines as the proof of \cref{prop:SRE_index} but on $\hat \cA$.

        \textit{(iv)} If $\omega_1$ and $\omega_2$ are symmetric invertible with inverses $ \omega'_1$ and $ \omega'_2$, then $\omega_1 \hat \otimes \omega_2$ is symmetric invertible on $\cA^1\hat\otimes \cA^2$ with inverse $ \omega'_1\hat \otimes  \omega'_2$. The defect state construction factors on each algebra. In particular, up to the reshuffling 
        $$
        \cA^1 \hat\otimes \cA^2 \hat \otimes (\cA^1)' \hat \otimes (\cA^2)' \cong \cA^1 \hat\otimes (\cA^1)' \hat \otimes   \cA^2 \hat \otimes (\cA^2)' = \hat \cA^1 \hat \otimes \hat \cA^2,
        $$
        one has $(\hat \omega_1 \hat \otimes \hat \omega_2)^D = (\hat \omega_1)^D \hat \otimes (\hat \omega_2)^D$ so that, by \cref{thm:PropertiesGeneralIndex}(viii),
        \begin{align*}
        \mathcal I(\omega_1 \hat\otimes \omega_2) &= \mathcal N_{\hat \rho^{Q_1}\hat\otimes \hat \rho^{Q_2}}(\omega_1 \hat\otimes \omega_2, (\hat \omega_1)^D \hat \otimes (\hat \omega_2)^D) \cr &= \mathcal N_{\hat \rho^{Q_1}}(\omega_1, (\hat \omega_1)^D) + \mathcal N_{\hat \rho^{Q_2}}(\omega_2, (\hat \omega_2)^D) \cr &= \mathcal I(\omega_1)+\mathcal I(\omega_2).
        \end{align*}
        
        \textit{(v)} This part is an immediate consequence of the quasi-free case, see \cref{thm:quasi-free} and \cref{prop:adiabatic vs spectral} below.
\end{proof}

\begin{rem}
    One may wonder if the abstract index could also find some relevance for the fractional quantum Hall effect. We believe that it may be the case in the following situation, which is similar to the quantization argument in~\cite{Categories}: Even if the initial state is not invertible, it may be that the flux insertion operation becomes inner after $q$ units of quantum flux have been inserted. If this is the case, the arguments above yield a fractional Hall conductance of the form $\frac{n}{q}$. It should be noted that in such a framework there is no prediction about the possible fractions one may obtain, see \cite{frohlich1991large}.
\end{rem}

%%%%%%%%%%%%%%%%%%%%%%%%%%%%%%%%%%%%%%
\section{Quasi-free fermions and invertibility}\label{sec:free}

%%%%%%%%%%%%%%%%%%%%%%%%%%%%%%%%%%%%%%%%%%%%%%%%%%%%%%%%%%%%%%%%%%%%%%%%%%%%%%%%%%%%%%%%%%%
%%%%%%%%%%%%%%%%%%%%%%%%%%%%%%%%%%%%%%%%%%%%%%%%%%%%%%%%%%%%%%%%%%%%%%%%%%%%%%%%%%%%%%%%%%%

We now return to the discussion at the beginning of~\cref{sec:flux_insertion} and we turn to the relation between the index of a pair of projections from \cref{thm:correspondence between many-body and single-particle indices} and the many-body index for the quantum Hall effect from \cref{thm:main}, when applied to quasi-free states. We shall show that the two indices coincide when they are simultaneously well-defined. 

Let $\mathcal H = \ell^2(\mathbb Z^2)$ (for simplicity, we ignore spin here) and let $P$ be the Fermi projection associated to some gapped local Hamiltonian $H$. Typically, we assume \eql{\label{eq:exp locality}|H_{x,y}|:= |\langle \delta_x, H \delta_y\rangle | \leq C \ee^{-\mu|x-y|},
} and the kernel of $P$ has the same property whenever the Fermi energy lies in a gap. Consider the quasi-free state $\omega_P$ on $\cA=\mathrm{CAR}(\mathcal H)$, see \cref{eq:quasifreefromP} above. We first recall 
\begin{prop}[{\cite[Proposition 2.5]{BachmannBolsRahnama24}}]\label{prop:Triviality of SRE}
    Assume $H$ is translation-invariant on $\mathcal H$ and consider the associated Chern number $c(P)$ for the Fermi projection. If $c(P)=0$ then $\omega_P$ is stably SRE. 
\end{prop}
Note that the assumption there is about time-reversal invariance, but the argument (in \cite[Appendix~B]{BachmannBolsRahnama24}) goes through the consequence thereof, which is that the Chern number vanishes.

Conversely, if a quasi-free state is stably SRE and the deformation is a family of quasi-free states, then the smooth path of states yields a smooth deformation of the bundle to a trivial bundle and so the Chern number must vanish.

Translation invariance is not expected to be relevant for this question. Indeed, \cite{chung2025essentiallycommutingunitary} proves that if $P$ and $Q$ have the same Chern number and both satisfy a certain non-triviality assumption (\cite[Definition 1.2]{chung2025essentiallycommutingunitary}), then there exists a path of local unitaries $[0,1]\ni t\mapsto U_t$ with $U_0=\Id$ such that $Q=U_1^\ast P U_1$, i.e., $P$ and $Q$ are homotopic within the space of local projections.

Now suppose $P$ has zero Chern number. By the same result, we can connect $P$ to a projection $Q$ that is diagonal in space (and hence also has zero Chern number). By lifting $t\mapsto U_t$ to an LGA, this yields a continuous path from $\omega_P$ to $\omega_Q$ (a product state), and therefore shows that $\omega_P$ is an SRE state.

The problem with this argument is that the notion of locality used in \cite{chung2025essentiallycommutingunitary} is a far cry from \cref{eq:exp locality}: there, $P$ is local iff $[P,L]\in\calK$, where $L$ is the Laughlin unitary $L\equiv\exp\br{\ii\arg\br{X_1+\ii X_2}}$. Bridging this gap would require formulating an analogous notion of locality in the many-body setting (and proving a Lieb--Robinson bound for it). Since this line of study is rather tangential to the present paper we do not pursue it further here and instead merely phrase a 
\begin{conj}\label{conjecture:Triviality of SRE}
    Let $H$ be a gapped Hamiltonian obeying \cref{eq:exp locality} for which the Fermi projection $P$ has a well defined topological index $c(P)$. Then $c(P)=0$ if and only if $\omega_P$ is stably SRE.
\end{conj}
\noindent Note that the `topological index' $c(P)$ is, strictly speaking, a `Chern number' only in the periodic setting, but we continue to use the same symbol for it.

We return to the notion of invertibility in the quasi-free setting. Let $\mathcal H \hat \oplus \mathcal H$ be the \emph{external} direct sum  and consider the projection $P\hat\oplus P^\perp$, where $P^\perp=\Id -P$. One then has $c(P\hat\oplus P^\perp) = c(P)+c(P^\perp)=0$. It follows that, in the translation-invariant setting, the state $\omega_{P\hat\oplus P^\perp} = \omega_P \hat \otimes \omega_{P^\perp}$ is stably-SRE. Consequently, the quasi-free state $\omega_P$ is invertible. In particular, we can associate to it the many-body index $\caI(\omega_P)$.

We use the independence of the index from the choice of parent Hamiltonian provided by \cref{thm:main}(ii) to choose $H\hat\oplus(-H)$ as parent Hamiltonian to $P\hat\oplus P^\perp$. Note that this direct sum structure is not what would arise from the general construction presented in \cref{lem:Q parent} because the automorphism mapping $\omega_{P\hat\oplus P^\perp}$ to a product state would in general mix the two layers. Now if the parent Hamiltonian is a direct sum, then the flux insertion unitary similarly factorizes as $U^\mathrm{qa}_{\phi_1,\phi_2} \hat\oplus U^\mathrm{qa}_{\phi_1,\phi_2}$. Here $U^\mathrm{qa}_{\phi_1,\phi_2}$ is the flux-insertion unitary given by
$$
U^\mathrm{qa}_{\phi_1,\phi_1}=\Id, \qquad \ii \partial_{\phi_2}U^\mathrm{qa}_{\phi_2,\phi_1} = K_{\phi_2} U^\mathrm{qa}_{\phi_2,\phi_1},
$$
where $K_\phi=\chi_\mathrm{left} K^\uparrow_\phi \chi_\mathrm{left}$ and $K^\uparrow_\phi= -\ii[\partial P^\uparrow_\phi, P^\uparrow_\phi]$ is the Kato generator of the family of projections $P^\uparrow_\phi = \ee^{\ii \phi \chi_\uparrow}P\ee^{-\ii \phi \chi_\uparrow}$. $\chi_\uparrow$ and $\chi_\mathrm{left}$ are the respective projections to the upper and left half-planes. See~\cite[Sec 5.1]{BachmannBolsRahnama24} for more details.

As a result, the defect state is again a quasi-free state, corresponding to a projection $P^\mathrm{qa}_{2\pi}\hat\oplus(P^\mathrm{qa}_{2\pi})^\perp$, where $P^\mathrm{qa}_{2\pi} = U^\mathrm{qa}_{0,2\pi} P (U^\mathrm{qa}_{0,2\pi})\str$. With this, the construction in the proof of \cref{thm:correspondence between many-body and single-particle indices} completely factorizes. However, in the final step of the proof when one computes the index, the charge automorphism $\rho$ must be taken with $Q\hat \otimes\Id$ rather than $Q\hat \otimes\Id + \Id\hat\otimes Q$, see the definition in \cref{eq:invertibleindex} and the discussion before it. Hence $\delta^\rho$ acts trivially on the second factor and we conclude

\begin{thm}\label{thm:quasi-free}
If the difference $P-P^\mathrm{qa}_{2\pi}$ is trace class then
$$
    \mathcal I(\omega_P) = \mathrm{index}(P, P^\mathrm{qa}_{2\pi}).
$$
\end{thm}

One may wonder when $P-P^\mathrm{qa}_{2\pi}$ is trace class, so that the index on the right hand side is well defined. A typical case is given by a local Hamiltonian as described above with Fermi level in a spectral gap\footnote{Note that if we assume that $\vert P_{x,y}\vert$ is exponentially decaying in $\vert x - y\vert$, then $-P \hat\oplus P$ is a bona fide parent Hamiltonian of $P \hat\oplus P^\perp$.} and \cref{prop:adiabatic vs spectral} below. Consider then  the instantaneous Fermi projection $P^{\mu}_\phi = \chi_{(-\infty, \mu]}({H}_{\phi})$ for the Hamiltonians ${H}_{\phi}$ defined by
\begin{equation} \label{eq:free fermion flux Hamiltonians}
    {H}_{\phi}(x, i ; y, j) = \begin{cases}
        \ep{\iu \phi \, \mathrm{sgn}( x_2 - y_2 )} {H}(x, i ; y, j) & \text{if } x_1, y_1 \leq 0 \\
        {H}(x, i; y, j) & \text{otherwise}
    \end{cases}
\end{equation}
Notice that, by construction, $H_{2\pi}=H_0$ so that $H_{2\pi}$ is gapped around $\mu$ as well, and moreover $P^{\mu}_{2\pi} = P^\mu_0 = P$.

\begin{prop}[{\cite[Proof of Prop. 5.1]{BachmannBolsRahnama24}}] \label{prop:adiabatic vs spectral}
    Let $\mu$ be in the spectral gap of $H_0$ and such that all eigenvalue crossings of $\phi \mapsto {H}_{\phi}$ across $\mu$ are simple. Then for all $\phi$ the difference $P^{\qa}_{\phi} - P^{\mu}_{\phi}$ is trace class and
    \begin{equation*}
        \mathrm{SF}_{\mu}( [0, 2\pi] \ni \phi \mapsto H_{\phi} ) = \Tr (P^{\qa}_{2\pi} - P^{\mu}_{2\pi}).
    \end{equation*}
\end{prop}
We recall that the spectral flow counts the signed number of eigenvalue crossings with the fiducial line at $\mu$. See for example~\cite{DeNittisSchulzBaldes2016spectral} for a detailed study of spectral flow in the single-particle and discrete setting for the quantum Hall effect.

Summarizing, under the hypothesis of \cref{prop:adiabatic vs spectral}, one has various single-particle interpretations of the index, 
$$
\mathcal I(\omega_P) = \mathrm{index}(P^\mu_{2\pi},P^\mathrm{qa}_{2\pi}) = \Tr(P^\mathrm{qa}_{2\pi}-P^\mu_{2\pi})= \mathrm{SF}_{\mu}( [0, 2\pi] \ni \phi \mapsto {H}_{\phi} ) 
$$
which ultimately appears as a many-body version of spectral flow induced by magnetic flux insertion.

\subsection{Beyond quasi-free states}

We now have a complete picture of the relationship between the many-body index of a pair of pure states, the index of a pair of projections, the spectral flow associated with flux insertion and the spectral flow in a non-interacting picture. We point out that the stability of the many-body index and many-body perturbation theory allow one to go beyond the free setting. The proposition below is phrased in a translation-invariant setting to be able to use \cref{prop:Triviality of SRE}. This assumption can be dropped if \cref{conjecture:Triviality of SRE} is true.
\begin{prop}
    Let $H_\lambda=H_0+\lambda V$ where $H_0$ and $V$ are two $Q$-preserving $0$-chains in $\mathcal A$. Assume that for all $\lambda \in [0,1]$, $H_\lambda$ has a locally unique gapped ground state $\omega_\lambda$, and that $\omega_0$ is a quasi-free translation-invariant state. Then  $\omega_{1}$ is invertible and $\mathcal I(\omega_{1})=\mathcal I (\omega_0)$. 
\end{prop}

\noindent Note that a locally unique gapped ground state is necessarily pure, see~\cite[Theorem A.3]{Tasaki}.

\begin{proof}
Since $H_\lambda$ has a uniform gapped ground state, the spectral flow~\cite{AutomorphicEq}, generated by a $0$-chain $K_\lambda$, is such that $\omega_{\lambda}= \omega_0\circ \alpha_{0\to\lambda}^K$. Moreover, since $\omega_0$ is quasi-free, it admits an inverse $\omega'_0$ on $\cA'$. Now,
$$
(\omega_1 \hat \otimes \omega'_0)\circ((\alpha_{0\to1}^K)^{-1} \hat \otimes \Id) = \omega_0 \hat \otimes \omega'_0,
$$
and since $\omega_0 \hat \otimes \omega'_0$ is stably SRE, this shows that $\omega_1$ is invertible. Finally, since $H_\lambda$ is $Q$-preserving, $\alpha^K$ is $Q$-preserving. Thus $\omega_1$ is $Q$-equivalent to $\omega_0$ and so $\mathcal I(\omega_{1})=\mathcal I (\omega_0)$ by \cref{thm:main}(iii). 
\end{proof}

	\bigskip
	\bigskip
	\noindent\textbf{Acknowledgments.}
	We are indebted to Alex Bols, Christopher Bourne, Gian Michele Graf and Anna Mazhar for stimulating discussions. JS is indebted to the CEREMADE at Universit\'e Paris Dauphine - PSL, as well as the Institute for Theoretical Physics at ETH for hospitality in the summer of 2025. 
	\bigskip

\bigskip

\noindent\textbf{Financial Support.} JS is supported in part by NSF grant DMS-2510207. SB is supported by NSERC of Canada.

\bigskip

\bibliographystyle{plain}
\bibliography{refs}

\end{document}